\documentclass[journal,onecolumn,11pt]{IEEEtran}
\usepackage{subfigure,cite,graphicx,amsmath,amssymb,eufrak,mathrsfs,epsf,epsfig}
\usepackage[usenames]{color}
\usepackage{psfrag}

\usepackage[active]{srcltx}
\ifCLASSINFOpdf

\else
\fi

\begin{document}
\newtheorem{ach}{Achievability}
\newtheorem{con}{Converse}
\newtheorem{definition}{Definition}
\newtheorem{theorem}{Theorem}
\newtheorem{lemma}{Lemma}
\newtheorem{example}{Example}
\newtheorem{cor}{Corollary}
\newtheorem{prop}{Proposition}
\newtheorem{conjecture}{Conjecture}
\newtheorem{remark}{Remark}
\title{Fundamental Limits of Spectrum Sharing Full-Duplex Multicell Networks}
\author{\IEEEauthorblockN{Sung Ho Chae,~\IEEEmembership{Member,~IEEE}, Sang-Woon Jeon,~\IEEEmembership{Member,~IEEE},\\and Sung Hoon Lim,~\IEEEmembership{Member,~IEEE}
}
\thanks{S. H. Chae is with Samsung Electronics, Suwon, South Korea (e-mail: sho.chae00@gmail.com).}
\thanks{S.-W. Jeon, the corresponding author, is with the Department of Information and Communication Engineering, Andong National University, Andong, South Korea (e-mail: swjeon@anu.ac.kr).}%
\thanks{S. H. Lim is with the School of Computer and Communication Sciences,
Ecole Polytechnique F{\'e}d{\'e}rale de Lausanne (EPFL), Lausanne,
Switzerland (e-mail: sung.lim@epfl.ch).}}%
 \maketitle

\begin{abstract}
This paper studies the degrees of freedom of full-duplex multicell networks that share the spectrum among multiple cells in a non-orthogonal setting. In the considered network, we assume that {\em full-duplex} base stations with multiple transmit and receive antennas communicate with multiple single-antenna mobile users. By spectrum sharing among multiple cells and (simultaneously) enabling full-duplex radio, the network can utilize the spectrum more flexibly, but, at the same time, the network is subject to multiple sources of interference compared to a network with separately dedicated bands for distinct cells and uplink--downlink traffic. Consequently, to take advantage of the additional freedom in utilizing the spectrum, interference management is a crucial ingredient.  
In this work, we propose a novel strategy based on interference alignment which takes into account inter-cell interference and intra-cell interference caused by spectrum sharing and full-duplex to establish a general achievability result on the sum degrees of freedom of the considered network.
Paired with an upper bound on the sum degrees of freedom, which is tight under certain conditions, we demonstrate how spectrum sharing and full-duplex can significantly improve the throughput over conventional cellular networks, especially for a network with large number of users and/or cells. 
\end{abstract}
\begin{IEEEkeywords}
Spectrum sharing, full-duplex, cellular networks, interference alignment, degrees of freedom
\end{IEEEkeywords}
 \IEEEpeerreviewmaketitle



\section{Introduction}
\PARstart{R}{adio} spectrum management of current state-of-the-art cellular networks are designed based on the {\em divide et impera} principle.
Due to practical capabilities of radio transceivers, regulatory requirements, standardization, and commercial reasons, radio spectrum has been divided in chunks of certain bandwidths. For example, the spectrum is divided between operators, between radio access technologies (e.g., GSM, LTE/A), between cells, and between uplink (UL) and downlink (DL) traffic (e.g., FDD, TDD). This principle may have provided a shortcut in the success of modern wireless communications, however, to support the ever increasing demands in wireless mobile data traffic, future cellular networks are required to use the limited spectrum more efficiently.

As the radio spectrum is a valuable resource, there have been several propositions for sharing the spectrum among multiple cells and operators, which in current design paradigms are competing entities that operate in different portions of frequency bands. By allowing {\em spectrum sharing}~\cite{SAPHYRE:02, Zhao:07, Akyildiz:08}, the spectrum is more flexibly utilized which can potentially increase the overall network spectral efficiency. Spectrum sharing has been considered in several aspects, e.g., in terms of radio access (orthogonal and non-orthogonal spectrum sharing~\cite{Lindblom:12}), in terms of operational range (single operator and intra-operator spectrum sharing~\cite{Ni:09, Gangula:13, Jorswieck:14}), and in terms of priority (primary--secondary sharing and equal priority (co-primary) spectrum sharing~\cite{Ahokangas:14, Hailu:14}).

In this paper, we study the fundamental limits of multicell non-orthogonal equal priority spectrum sharing networks. In particular, we take an {\em information theoretic approach} and model the network as an interference network with multiple base stations (BSs) and mobile users, each BS equipped with multiple antennas and each mobile user equipped with a single antenna. Indeed, there have been several approaches by modeling a spectrum sharing network as an interference channel, either in primary--secondary spectrum sharing scenarios (in the context of cognitive radio)~\cite{Devroye:06, Goldsmith:09} or in equal priority scenarios~\cite{Lindblom:12, Lindblom:14}. In our work, we further explore this direction by considering a generalization of the interference channel to represent a {\em multi-cell} network and study the degrees of freedom (DoF) in such networks. By analyzing the sum DoF, we approximately characterize the sum capacity of the considered network in the high signal-to-noise ratio regime.

Another important distinction in our work is that we consider multiantenna {\em full-duplex (FD) BSs}.
Results in~\cite{Choi10, Aryafar12,Khandani13, Duarte10, Jain11, Bharadia13} have demonstrated the feasibility of FD wireless communication by suppressing or cancelling self-interference in the RF and baseband domain. By enabling FD radio, simultaneous transmission and reception can potentially double the spectral efficiency of current half-duplex (HD) systems~\cite{Hong13}. In an alternative view, FD can be considered as an extreme case of non-orthogonal spectrum sharing among UL and DL traffic.
In this sense, our work studies two spectrum sharing concepts in a single comprehensive framework. The main challenge in this direction is to deal with {\em interference} caused by several coexisting entities and simultaneous non-orthogonal transmissions among them. This involves sufficient treatment of inter-cell interference, intra-cell interference, and additionally, interference caused by UL and DL transmissions.
The goal of this paper is to present a theoretical framework for understanding the DoF in FD multicell spectrum sharing networks and help understand the overall system improvement by unlocking the spectrum barriers.

\subsection{Related Work}
To deal with the additional interference caused by spectrum sharing among multiple cells and FD, we propose signal space \emph{interference alignment} (IA) strategies that are optimized for multicell FD networks. The pioneering work of Maddah-Ali, Motahari, and Khandani~\cite{Maddah-Ali:08} and Cadambe and Jafar~\cite{Jafar08,Cadambe107} introduced IA as an effective coding technique that efficiently deals with interference. In particular, IA has been shown to achieve the optimal DoF for various interference networks~\cite{Suh08,Suh:11,Viveck1:09,Viveck2:09,Tiangao:10,Annapureddy:11,Ke:12,Tiangao:12,Jeon4:12,JeonSuh:14,Sahai13,Nazer11:09}.

When applied to {\em cellular networks}, it was shown that IA can be successfully applied to mitigate interference in a two-cell network~\cite{Suh08,Suh:11}. Motivated by this approach, the work of~\cite{JeonSuh:14} proposed an IA strategy for a spectrum sharing multicell network with dynamic UL--DL configuration and showed that it attains the optimal DoF. Following this line of work, we have previously studied and characterized the optimal DoF of a {\em single cell} network that consists of a FD BS and, either, HD mobile users or FD mobile users~\cite{Jeon--Chae--Lim2015_ISIT}. The DoF improvement achievable by FD operation at BSs has been also studied in the context of blind IA techniques using reconfigurable antennas \cite{Yang:16} and linear beamforming techniques for multiple-input and multiple-output (MIMO) environment \cite{Kim:15}. In~\cite{Jeon--Chae--Lim2015_ISIT}, the UL data is sent to the BS using IA such that user-to-user interference is confined within a tolerated number of signal dimensions, while the BS transmits in the remaining signal dimensions via zero-forcing beamforming for the DL transmission. This work is a follow-up work on~\cite{Jeon--Chae--Lim2015_ISIT} by further extending the network to a spectrum shared network, i.e., a multicell configuration.

\subsection{Contributions and Organization}
Compared to single-cell cellular networks (where cells are assumed to operate in different frequency bands), multicell spectrum sharing inherents additional interference from cells that share the same spectrum. Thus, it is not clear whether spectrum sharing among multiple cells will provide some gain over a cellular network with cells that operate in separately dedicated bands. Assuming $K$ is the number of cells that share the same spectrum, what is the DoF of this network? How does the DoF behave with respect to $K$?

To answer these questions, we propose a novel strategy based on IA which takes into account inter-cell interference, intra-cell interference, as well as interference caused by FD. The key idea of our proposed scheme is to minimize the dimension of UL inter-cell and user-to-user interferences via IA beamforming at UL users, and at the same time, to align or null out DL intra-cell and inter-cell interferences, as well as BS-to-BS interference via IA beamforming or zero-forcing beamforming at BSs. In the proposed strategy, each BS employs IA beamforming when the number of antennas at each BS is small while zero-forcing beamforming is used otherwise. The DoF performance of this strategy is stated in Theorem~\ref{thm:dof_1}. Furthermore, we also derive an upper bound on the sum DoF (Theorem~\ref{thm:dof_2}), which is tight under certain conditions. These results demonstrate how spectrum sharing and FD can provide significant throughput improvement over conventional cellular networks, especially for a network with large number of users and/or cells. These gains are highlighted in Examples~\ref{ex:ex1} and \ref{ex:ex2} in Section~\ref{sec:main-results}.


Starting with the next section, we describe the network model and the sum DoF metric considered in this paper. In Section~\ref{sec:main-results}, we present the main results of the paper and intuitively explain how multicell non-orthogonal spectrum sharing and FD operation can increase the overall DoF. In Section~\ref{sec:achievability} we provide the achievability proof of Theorem~\ref{thm:dof_1}. Finally, in Section~\ref{sec:discussion} we briefly discuss about the impacts of BS-to-BS interference and self-interference on DoF.

We will use boldface lowercase letters to denote vectors and boldface uppercase letters to denote matrices.
Throughout the paper, $[1:n]$ denotes $\{1,2,\cdots,n\}$, $\mathbf{0}_n$ denotes the $n\times 1$ all-zero vector, $\mathbf{0}_{m\times n}$ denotes the $m\times n$ all-zero matrix, and $\mathbf{I}_n$ denotes the $n\times n$ identity matrix.
For a real value $a$, $\lceil a \rceil$ is the smallest integer greater than or equal to $a$.
For a set of vectors $\{\mathbf{a}_i\}$, $\operatorname{span}(\{\mathbf{a}_i\})$ denotes the vector space spanned by the vectors in $\{\mathbf{a}_i\}$.
For a vector $\mathbf{b}$, $\mathbf{b}\perp\operatorname{span}(\{\mathbf{a}_i\})$ means that $\mathbf{b}$ is orthogonal with all vectors in $\operatorname{span}(\{\mathbf{a}_i\})$.
For a matrix $\mathbf{A}$, $\mathbf{A}^{\dagger}$ and $\mathbf{A}^{-1}$ denote the transpose and inverse of $\mathbf{A}$, respectively.
For a set of matrices $\{\mathbf{A}_i\}$, $\operatorname{diag}(\mathbf{A}_1,\cdots, \mathbf{A}_n)$ denotes the block diagonal matrix consisting of $\{\mathbf{A}_i\}$.

\section{Problem Formulation} \label{sec:problem}
In this section, we formally introduce the network model and performance metric used in the paper.

\subsection{Network Model}
Consider a $K$-cell cellular network consisting of FD BSs and HD UL and DL users. 
Each BS $k\in[1:K]$ (the BS in the $k$th cell) wishes to send a set of independent messages $(W^{[\sf d]}_{k1},W^{[\sf d]}_{k2}, \cdots,W^{[\sf d]}_{kN})$ to its $N$ DL users and at the same time wishes to receive a set of independent messages $(W^{[\sf u]}_{k1},W^{[\sf u]}_{k2}, \cdots,W^{[\sf u]}_{kN})$ from its $N$ UL users in the same cell.
We assume that each BS is equipped with $M$ transmit antennas and $M$ receive antennas while each user is equipped with a single antenna. 
For simplicity, denote the $j$th UL user and the $j$th DL user in the $k$th cell by UL user $(k,j)$ and DL user $(k,j)$, respectively.

Let $\mathbf{g}_{ij,k}[t]\in \mathbb{R}^{1\times M}$ be the channel vector from BS $k$ to DL user $(i,j)$ at time $t$, $\mathbf{f}_{i,jk}[t]\in \mathbb{R}^{M\times 1}$ be the channel vector from UL user $(j,k)$ to BS $i$ at time $t$, $h_{ij,kl}[t]\in \mathbb{R}$ be the scalar channel from UL user $(k,l)$ to DL user $(i,j)$ at time $t$, and $\mathbf{B}_{ij}[t]\in \mathbb{R}^{M\times M}$ be the channel matrix from BS $j$ to BS $i$ at time $t$.
We assume that self-interference within each BS is perfectly suppressed so that $\mathbf{B}_{ii}[t]=\mathbf{0}_{M\times M}$ for all $i\in[1:K]$.
We further assume that all channel coefficients (except $\mathbf{B}_{ii}[t]$) are independent and identically distributed (i.i.d.) drawn from a continuous distribution and vary independently over time.
Global channel state information (CSI) is assumed to be available at each BS and each UL and DL user.
The assumption on global CSI is somewhat an idealistic assumption in practice. Nonetheless, we remark that the fundamental performance of spectrum sharing networks, even under this idealistic assumption, is unknown. The scope of this work is to provide an initial step in this direction for further studies based on more realistic CSI assumptions.

For $i\in[1:K]$ and $j\in[1:N]$, the received signal of DL user $(i,j)$ at time $t$, denoted by $y_{ij}^{[\sf d]}[t]\in \mathbb{R}$, is given by
\begin{align}\label{channel model1}
y^{[\sf d]}_{ij}[t]&=\sum_{k=1}^{K}\mathbf{{g}}_{ij,k}[t]{\mathbf{x}}^{[\sf bs]}_{k}[t]+\sum_{k=1}^{K}\sum_{l=1}^{N}h_{ij,kl}[t]x^{[\sf u]}_{kl}[t]+z_{ij}^{[\sf d]}[t]
\end{align}
and the received signal vector of BS $i$ at time $t$, denoted by $\mathbf{y}_{i}[t]^{[\sf bs]}\in \mathbb{R}^{M\times 1}$, is given by 
\begin{align}\label{channel model2}
{\mathbf{y}}^{[\sf bs]}_{i}[t]&=\sum_{j=1}^{K}\sum_{k=1}^{N}\mathbf{{f}}_{i,jk}[t]x^{[\sf u]}_{jk}[t]+\sum_{j=1}^K\mathbf{{B}}_{ij}[t]{\mathbf{x}}^{[\sf bs]}_{j}[t]+\mathbf{{z}}^{[\sf bs]}_i[t],
\end{align}
where $\mathbf{x}_k^{[{\sf bs}]}[t]\in\mathbb{R}^{M\times1}$ is the transmit signal vector of BS $k$ at time $t$, $x^{[{\sf u}]}_{kl}[t]\in\mathbb{R}$ is the transmit signal of UL user $(k,l)$ at time $t$, $\mathbf{z}_i^{[{\sf bs}]}[t]\in\mathbb{R}^{M\times 1}$ is the additive noise vector of BS $k$ at time $t$, and $z_{ij}^{[{\sf d}]}[t]\in\mathbb{R}$ is the additive noise of DL user $(i,j)$ at time $t$.
Each BS and each UL user should satisfy an average transmit power constraint, i.e., $\frac{1}{n}\sum_{t=1}^n\|\mathbf{x}_i^{[{\sf bs}]}[t]\|^2\leq P$ and $\frac{1}{n}\sum_{t=1}^n(x_{ij}^{[{\sf u}]}[t])^2\leq P$ for all $i\in[1:K]$ and $j\in[1:N]$.
We assume that all elements in $\mathbf{z}_i^{[{\sf bs}]}[t]$ and $z_{ij}^{[{\sf d}]}[t]$ are i.i.d. drawn from $\mathcal{N}(0,1)$.

Throughput the paper, we will also use the following definitions:
\begin{align}
\mathbf{{B}}_{ij}[t]&=\left[\begin{array}{cccc}\mathbf{b}_{i,j1}[t],&\mathbf{b}_{i,j2}[t],&\cdots,&\mathbf{b}_{i,jM}[t]\end{array}\right]\nonumber\\
&=\left[\begin{array}{cccc}b_{i1,j1}[t],&b_{i1,j2}[t],&\cdots,&b_{i1,jM}[t]\nonumber\\
b_{i2,j1}[t],&b_{i2,j2}[t],&\cdots,&b_{i2,jM}[t]\nonumber\\
\vdots&\vdots&\ddots&\vdots\\
b_{iM,j1}[t],&b_{iM,j2}[t],&\cdots,&b_{iM,jM}[t]\end{array}\right],\nonumber\\
\mathbf{g}_{ij,k}[t]&=\left[\begin{array}{cccc}g_{ij,k1}[t],&g_{ij,k2}[t],&\cdots,&g_{ij,kM}[t]\end{array}\right],\nonumber\\
\mathbf{f}_{i,jk}[t]&=\left[\begin{array}{cccc}f_{i1,jk}[t],&f_{i2,jk}[t],&\cdots,&f_{iM,jk}[t]\end{array}\right]^{\dagger}, 
\end{align}
where $\mathbf{b}_{i,jk}[t]\in \mathbb{R}^{M\times 1}$ is the channel vector from the $k$th transmit antenna of BS $j$ to BS $i$ at time $t$, $b_{ik,jl}[t]\in \mathbb{R}$ is the scalar channel from the $l$th transmit antenna of BS $j$ to the $k$th receive antenna of BS $i$ at time $t$, $g_{ij,kl}[t]\in \mathbb{R}$ is the scalar channel from the $l$th transmit antenna of BS $k$ to DL user $(i,j)$ at time $t$, and $f_{il,jk}[t]\in \mathbb{R}$ is the scalar channel from UL user $(j,k)$ to the $l$th receive antenna of BS $i$ at time $t$.

\begin{figure}[t!]
\begin{center}
\includegraphics[width=0.45\textwidth]{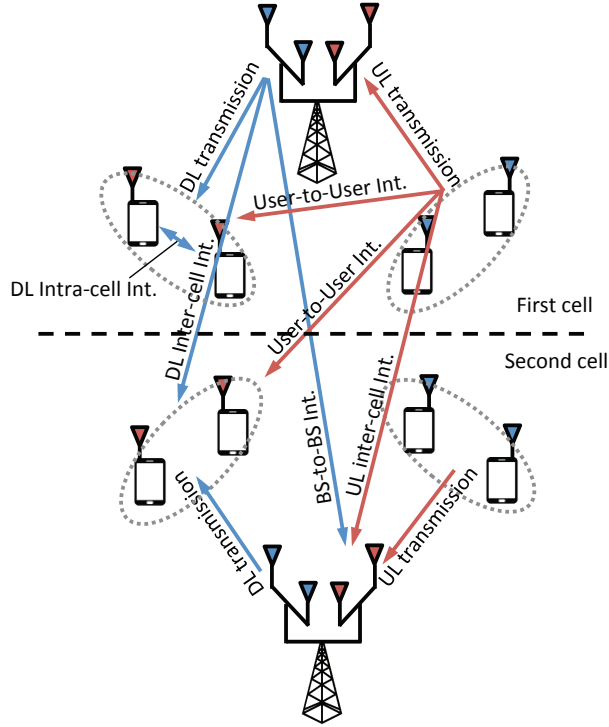}
\end{center}
\caption{Types of Interferences for the $(2,2,2)$ FD cellular network.}
\label{Fig:Channel_model}
\end{figure}

For here and henceforth, we refer to the above network by the $(K,M,N)$ FD cellular network. As an example, Fig.~\ref{Fig:Channel_model} illustrates types of interferences based on the $(2,2,2)$ FD cellular network. 
For simplicity, we omit the interferences caused by the second cell in the figure.
As seen in the figure, each DL user suffers from DL intra-cell, DL inter-cell, and user-to-user interferences, and each BS suffers from UL intra-cell, UL inter-cell, and BS-to-BS interferences. It is worthwhile to mention that, compared to HD networks, user-to-user and BS-to-BS interferences are exclusive interference components caused by enabling FD operation at the BSs. In addition, compared to the single-cell FD case, UL and DL inter-cell interferences and BS-to-BS interferences are exclusive interference components caused by enabling multicell spectrum sharing.

\subsection{Degrees of Freedom}

Let $W^{[{\sf d}]}_{ki}$ and $W^{[{\sf u}]}_{kj}$ be chosen uniformly at random from $[1:2^{nR^{[{\sf d}]}_{ki}}]$ and $[1:2^{nR^{[{\sf u}]}_{kj}}]$ respectively, where $i,j\in[1:N]$ and $k\in[1:K]$.
Then a length-$n$  $(2^{nR^{[{\sf d}]}_{11}},\cdots,2^{nR^{[{\sf d}]}_{KN}},2^{nR^{[{\sf u}]}_{11}},\cdots,2^{nR^{[{\sf u}]}_{KN}};n)$ code consists of the following set of encoding and decoding functions:
\begin{itemize}
\item {\em Encoding:} For $t\in[1:n]$, the encoding function of BS $k$ at time $t$ is given by 
\begin{align*}
\mathbf{x}_k^{[\sf bs]}[t]=\phi_t(W_{k1}^{[{\sf d}]},\cdots,W_{kN}^{[{\sf d}]},\mathbf{y}^{[{\sf bs}]}_k[1],\cdots \mathbf{y}_k^{[{\sf bs}]}[t-1]).
\end{align*}
For $t\in[1:n]$, the encoding function of UL user $(k,j)$ at time $t$ is given by 
\begin{align*}
x_{kj}[t]=\varphi_t(W_{kj}^{[{\sf u}]}),
\end{align*}
where $k\in[1:K]$ and $j\in[1:N]$.
\item {\em Decoding:} Upon receiving $\mathbf{y}_k^{[{\sf bs}]}[1]$ to $\mathbf{y}_k^{[{\sf bs}]}[n]$, the decoding function of BS $k$ is given by 
\begin{align*}
	\hat{W}^{[{\sf u}]}_{kj}=\chi_{kj}(\mathbf{y}^{[{\sf bs}]}_k[1],\cdots,\mathbf{y}_k^{[{\sf bs}]}[n],W_{k1}^{[{\sf d}]},\cdots,W_{kN}^{[{\sf d}]}) \text{ for } j\in[1:N]. 
\end{align*}
Upon receiving $y_{ki}^{[{\sf d}]}[1]$ to $y_{ki}^{[{\sf d}]}[n]$, the decoding function of DL user $(k,i)$ is given by 
\begin{align*}
\hat{W}^{[{\sf d}]}_{ki}=\psi_{ki}(y_{ki}^{[{\sf d}]}[1],\cdots,y_{ki}^{[{\sf d}]}[n]),
\end{align*} 
where $k\in[1:K]$ and $i\in[1:N]$.
\end{itemize}

A rate tuple $(R^{[{\sf d}]}_{11},\cdots,R^{[{\sf d}]}_{KN},R^{[{\sf u}]}_{11},\cdots,R^{[{\sf u}]}_{KN})$ is said to be \emph{achievable} for the $(K,M,N)$ FD cellular network if there exists a sequence of  $(2^{nR^{[{\sf d}]}_{11}},\cdots,2^{nR^{[{\sf d}]}_{KN}},2^{nR^{[{\sf u}]}_{11}},\cdots,2^{nR^{[{\sf u}]}_{KN}};n)$ codes such that $\Pr(\hat{W}^{[{\sf d}]}_{ki}\neq W^{[{\sf d}]}_{ki})\to 0$ and $\Pr(\hat{W}^{[{\sf u}]}_{kj}\neq W^{[{\sf u}]}_{kj})\to 0$ as $n$ increases for all $i,j\in[1:N]$ and $k\in[1:K]$.
Then a DoF tuple is said to be \emph{achievable} if 
\begin{align}
(d^{[{\sf d}]}_{11},\cdots,d_{KN}^{[{\sf d}]},d^{[{\sf u}]}_{11},\cdots,d_{KN}^{[{\sf u}]})=\lim_{P\to\infty}\left(\frac{R^{[{\sf d}]}_{11}}{\frac{1}{2}\log P},\cdots,\frac{R^{[{\sf d}]}_{KN}}{\frac{1}{2}\log P},\frac{R^{[{\sf u}]}_{11}}{\frac{1}{2}\log P},\cdots,\frac{R^{[{\sf u}]}_{KN}}{\frac{1}{2}\log P}\right)
\end{align}
for some achievable rate tuple $(R^{[{\sf d}]}_{11},\cdots,R^{[{\sf d}]}_{KN},R^{[{\sf u}]}_{11},\cdots,R^{[{\sf u}]}_{KN})$.\footnote{For complex channels, $(d^{[{\sf d}]}_{11},\cdots,d_{KN}^{[{\sf d}]},d^{[{\sf u}]}_{11},\cdots,d_{KN}^{[{\sf u}]})$ is defined as $\lim_{P\to\infty}\left(\frac{R^{[{\sf d}]}_{11}}{\log P},\cdots,\frac{R^{[{\sf d}]}_{KN}}{\log P},\frac{R^{[{\sf u}]}_{11}}{\log P},\cdots,\frac{R^{[{\sf u}]}_{KN}}{\log P}\right)$. For this case, we can achieve the same DoF tuple by applying a coding scheme proposed for real channels to both the real and imaginary parts.}

We further denote the maximum achievable sum DoF of the $(K,M,N)$ FD cellular network by $d_{\Sigma}$, i.e.,
\begin{align}
d_{\Sigma}=\max_{(d^{[{\sf d}]}_{11},\cdots,d_{KN}^{[{\sf d}]},d^{[{\sf u}]}_{11},\cdots,d_{KN}^{[{\sf u}]})\in\mathcal{D}}\left\{\sum_{k=1}^K\sum_{i=1}^{N}\left(d_{ki}^{[{\sf d}]}+d_{ki}^{[{\sf u}]}\right)\right\},
\end{align}
where $\mathcal{D}$ denotes the DoF region of the $(K,M,N)$ FD cellular network. For the rest of the paper, we will focus on the sum DoF of the $(K,M,N)$ FD cellular network.

\section{Main Results} \label{sec:main-results}

In this section, we first describe our main results, the sum DoF of the $(K,M,N)$ FD cellular network, and then compare with the conventional HD cellular network.
In the following, we establish an achievable lower bound on $d_{\Sigma}$.

\begin{theorem}\label{thm:dof_1}
For the $(K,M,N)$ FD cellular network, the following sum DoF is achievable:
\begin{align}\label{dof_achievable_theorem}
&d_{\Sigma}\geq \left\{\begin{array}{lll}
KN\left(\frac{M^2+MN}{M^2+N^2+MN}\right)~~\quad\quad\quad\quad\quad\quad {~}~\text{if $M\leq (K-2)N$, }\\
\max\left\{KN\left(\frac{M^2+MN}{M^2+N^2+MN}\right),\frac{2MN+M^2}{M+N},\min\left\{\frac{MK}{K-1}, (K-1)N\right\}\right\}  \\ 
\quad\quad\quad\quad\quad\quad\quad\quad\quad\quad\quad\quad\quad\quad\quad\quad\text{if $(K-2)N <M< (K-1)N$,}\\
M+\frac{M}{K}\quad \quad\quad\quad\quad\quad\quad\quad\quad\quad\quad\quad~~\text{if $(K-1)N\leq M<\frac{K^2N}{K+1}$,}\\
KN~~~\quad\quad\quad\quad\quad\quad\quad\quad\quad\quad\quad\quad\quad \text{otherwise.}
\end{array} \right.
\end{align}
\end{theorem}
\begin{proof}
The proof of Theorem \ref{thm:dof_1} is deferred to Section \ref{sec:achievability}.
\end{proof}


\begin{remark}[FD-user case]
Following a similar argument in\cite[Theorem 2]{Jeon--Chae--Lim2015_ISIT}, it is obvious that the achievable lower bound on $d_{\Sigma}$ in Theorem \ref{thm:dof_1} is also valid for the FD-user case consisting of $N$ FD users in each cell (instead of $N$ HD UL users and $N$ HD DL users), simultaneously sending and receiving UL and DL messages with the FD BS in the same cell. \hfill$\lozenge$
\end{remark}

\begin{remark}[Single-cell case] \label{re:single_cell}
For the single-cell case, i.e., $K=1$, it has been shown independently in~\cite{Jeon--Chae--Lim2015_ISIT, Bai_ISIT} that  $d_{\Sigma}=\min\{2M,N\}$, which coincides with the achievable sum DoF in Theorem \ref{thm:dof_1}. \hfill$\lozenge$
\end{remark}

In order to see the DoF gain from the FD operation at BSs, we introduce the sum DoF of the $(K,M,N)$ HD cellular network in which each HD BS equipped with $M$ transmit antennas supports $N$ HD DL users in the same cell. For this case, the optimal sum DoF, denoted by $d_{\Sigma, HD}$, has been characterized in \cite[Theorem 3.4]{Sridharan:13}.

\begin{theorem}[Sridharan--Yu \cite{Sridharan:13}] \label{thm:HD_case}
For the $(K,M,N)$ HD cellular network,
\begin{align}\label{eq:HD dof}
&d_{\Sigma,HD}=\left\{\begin{array}{lll} \frac{KMN}{M+N}{~~} \text{if $M<(K-1)N$,}\\
M \quad\quad \text{if $(K-1)N\leq M<KN$,}\\
KN~~~ ~\text{otherwise.}
\end{array} \right.
\end{align}
\end{theorem}

Although we assume DL transmission for the $(K,M,N)$ HD cellular network, Theorem \ref{thm:HD_case} still holds if we define the $(K,M,N)$ HD cellular network based on UL transmission.

We also establish an upper bound on $d_{\Sigma}$ for $K\geq 2$ in the following theorem, which can characterize $d_{\Sigma}$ for a class of network topologies.
\begin{theorem}\label{thm:dof_2}
Consider the $(K,M,N)$ FD cellular network. For $K\geq 2$, any achievable sum DoF should satisfy the following upper bound:
\begin{align}\label{dof_converse theorem}
d_{\Sigma}\leq  K\min\{M,N\}.
\end{align}
\end{theorem}
\begin{proof}
For $i,k\in[1:K]$ satisfying $i\neq k$, we eliminate all the messages except 
\begin{align}
(W_{i1}^{\sf[d]},\ldots, W_{iN}^{\sf[d]}, W_{k1}^{\sf[u]},\ldots, W_{kN}^{\sf[u]}),
\end{align}
which does not decrease $\sum_{j=1}^{N}d_{ij}^{[\sf d]}+\sum_{j=1}^{N}d_{kj}^{[\sf u]}$. Then, from the result in~\cite[Theorem 1]{JeonSuh:14}, we have $\sum_{j=1}^{N}d_{ij}^{[\sf d]}+\sum_{j=1}^{N}d_{kj}^{[\sf u]}\leq \min\{M,N\}$. Hence, by summing up all such bounds for $i,k\in[1:K]$ satisfying $i\neq k$, we have 
\begin{align}\label{eq:converse 1}
\sum_{i=1}^{K}\sum_{j=1}^{N}d_{ij}^{[\sf d]}+\sum_{k=1}^{K}\sum_{j=1}^{N}d_{kj}^{[\sf u]}\leq K\min\{M,N\}, 
\end{align}
which completes the proof.
\end{proof}
\begin{remark}
As another approach to obtain an upper bound on the sum DoF, one may consider the case in which all the DL messages (or UL messages) are eliminated, which does not decrease 
$\sum_{i=1}^{K}\sum_{j=1}^{N}d_{ij}^{[\sf u]}$ (or $\sum_{i=1}^{K}\sum_{j=1}^{N}d_{ij}^{[\sf d]}$). Then, from Theorem \ref{thm:HD_case}, we have $\sum_{i=1}^{K}\sum_{j=1}^{N}d_{ij}^{[\sf u]}\leq d_{\Sigma, HD}$ (or $\sum_{i=1}^{K}\sum_{j=1}^{N}d_{ij}^{[\sf d]}\leq d_{\Sigma, HD}$), 
which implies that
\begin{align}\label{eq:converse 2}
d_{\Sigma}\leq 2d_{\Sigma, HD}.
\end{align}
However, it can be easily checked that the upper bound~\eqref{dof_converse theorem} in Theorem~\ref{thm:dof_2} is tighter than~\eqref{eq:converse 2} for  all values of $K$, $M$, and $N$.\hfill$\lozenge$
\end{remark}


\begin{remark}[Optimality condition]
As mentioned in Remark \ref{re:single_cell}, $d_{\Sigma}$ has been characterized for the single-cell case \cite{Jeon--Chae--Lim2015_ISIT, Bai_ISIT} .
Hence, focus on the multi-cell case where $K\geq 2$. 
By comparing the lower and upper bounds in Theorems \ref{thm:dof_1} and \ref{thm:dof_2}, it can be seen that $d_{\Sigma}=KN$ when $M\geq \frac{K^2N}{K+1}$. In addition, $d_{\Sigma}$ converges to $KM$ as $N$ increases, demonstrating that the lower and upper bounds in  Theorems \ref{thm:dof_1} and \ref{thm:dof_2} are asymptotically tight in the limit of large $N$.
\hfill$\lozenge$
\end{remark}

Let us now compare the sum DoFs of the FD and HD cellular networks.
It is easily seen that the achievable sum DoF in Theorem \ref{thm:dof_1} is strictly greater than that in Theorem \ref{thm:HD_case} if $M<KN$. In particular, as $K$ increases, the multiplicative gap between $d_{\Sigma}$ and $d_{\Sigma, HD}$ is lower bounded by 
\begin{align}
\frac{d_{\Sigma}}{d_{\Sigma, HD}}\geq 1+\frac{MN}{M^2+N^2+MN},
\end{align}
from the first cases in Theorems \ref{thm:dof_1} and \ref{thm:HD_case}.
Therefore, \emph{the additive gap between $d_{\Sigma}$ and $d_{\Sigma, HD}$ becomes arbitrarily large as $K$ increases.}
Furthermore, for both FD and HD cases, the sum DoFs increase linearly proportionally with increasing $K$, showing that \emph{the gain from spectrum sharing among cells is significant} in spite of various interference sources as seen in Fig. \ref{Fig:Channel_model}.


\begin{figure}[t!]
\begin{center}
\includegraphics[width=0.55\textwidth]{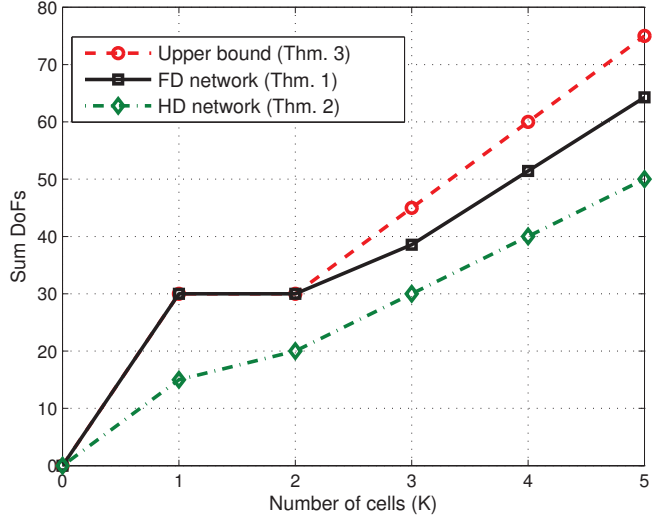}
\end{center}
\vspace{-0.2in}
\caption{Sum DoFs with respect to $K$ when $M=15$ and $N=30$.}
\vspace{-0.1in}
\label{Fig:sum dof with K}
\end{figure}

\begin{example}[DoF comparison with respect to $K$] \label{ex:ex1}
Figure~\ref{Fig:sum dof with K} plots the sum DoFs of the FD and HD cellular networks with respect to $K$ when $M=15$ and $N=30$. It can be seen that the sum DoFs increase with respect to $K$, showing the gain from spectrum sharing among cells. Furthermore, FD operation at BSs strictly improves the sum DoF compared to the conventional HD operation at BSs for all $K\geq 1$ and, moreover, the additive sum DoF gap increases as $K$ increases.  \hfill$\lozenge$
\end{example}

\begin{figure}[t!]
\begin{center}
\includegraphics[width=0.55\textwidth]{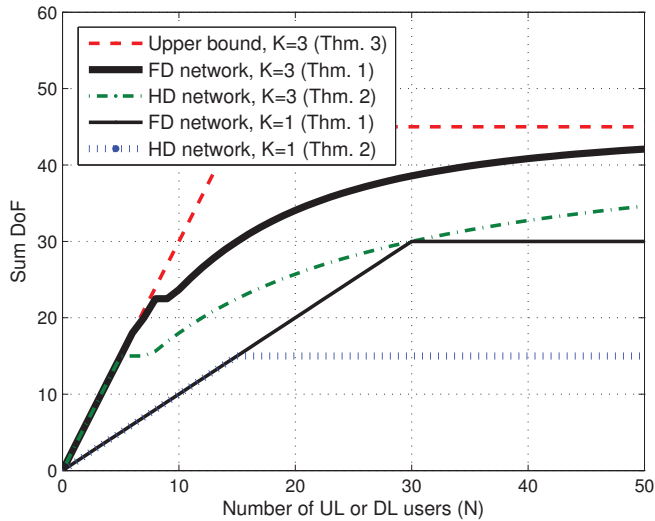}
\end{center}
\vspace{-0.2in}
\caption{Sum DoFs with respect to $N$ when $M=15$.}
\vspace{-0.1in}
\label{Fig:sum dof with N}
\end{figure}

\begin{example}[DoF comparison with respect to $N$]\label{ex:ex2}
Figure~\ref{Fig:sum dof with N} plots the sum DoFs of the FD and HD cellular networks with respect to $N$ when $K=3$ and $M=15$. For comparison, we also plot the sum DoFs for the single-cell case, i.e., $K=1$ and $M=15$.
Again, by comparing $K=3$ and $K=1$ cases, we can observe the sum DoF improvement from spectrum sharing among cells and FD operation at BSs. In particular, FD operation at BSs can strictly improve the sum DoF compared to the HD case if $N\geq \frac{M}{K}$.   \hfill$\lozenge$
\end{example}

\section{Achievability}\label{sec:achievability}

In this section, we prove Theorem \ref{thm:dof_1}.
For $M\geq KN$, HD operation at BSs is enough to achieve $d_{\Sigma}=KN$ from Theorem \ref{thm:HD_case}.
Therefore, we focus on the case where $M<KN$ in this section. 
Unlike the single-cell case~\cite{Jeon--Chae--Lim2015_ISIT}, we need to take into account for UL and DL inter-cell interferences as well as BS-to-BS interference for the multi-cell case as seen in Fig. \ref{Fig:Channel_model}.
In the following we propose two interference management schemes and show that the maximum sum DoF achievable by two proposed schemes coincides with the lower bound on $d_{\Sigma}$ in Theorem \ref{thm:dof_1}.

\subsection{Scheme 1: UL and DL Interference Alignment} \label{subsec:scheme1}
The first scheme applies 1) UL IA performed by UL users for aligning UL inter-cell and user-to-user interferences and 2) DL IA performed by BSs for aligning DL intra-cell, DL inter-cell, and BS-to-BS interferences. More specifically, each transmit and receive antenna of BSs is treated as a virtual user to apply DL IA.
We will show that, for $M<KN$, the following sum DoF is achievable:
\begin{align}\label{eq:Dof1}
d_{\Sigma,1}=\max_{\substack{
            \lambda_1,\lambda_2\in(0,1]\\
            \lambda_1(N+\min\{M,K(N-1)\}+\lambda_2N\leq M\\
           \lambda_1+\lambda_2\left(1+\frac{N}{M}\right)\leq 1\\
           }}KN(\lambda_1+\lambda_2).
            \end{align}


In order to establish \eqref{eq:Dof1}, assume that $\lambda_1$ and $\lambda_2$ satisfy the three constraints in  \eqref{eq:Dof1} from now on.
Define $\mathcal{S}_T^{[\sf u]}=[0:T-1]^{KN(KN+KM)}$ and $\mathcal{S}_T^{[\sf d]}=[0:T-1]^{KM(KN+KM)}$, where $T$ will be specified later on. For UL messages, we divide $W^{[\sf u]}_{ij}$, $i\in[1:K]$ and $j\in[1:N]$, into $T^{KN(KN+KM)}$ submessages $\{W^{[\sf u](\mathbf{s}^{[\sf u]})}_{ij}\}_{\mathbf{s}^{[\sf u]}\in \mathcal{S}_T^{[\sf u]}}$. A length-$n$ Gaussian codeword associated with $W^{[\sf u](\mathbf{s}^{[\sf u]})}_{ij}$ is denoted by
$\big[\begin{array}{cccc}c^{[\sf u](\mathbf{s}^{[\sf u]})}_{ij}[1],& c^{[\sf u](\mathbf{s}^{[\sf u]})}_{ij}[2], &\cdots, & c^{[\sf u](\mathbf{s}^{[\sf u]})}_{ij}[n]\end{array}\big]$,
where its coefficients are generated i.i.d. from $\mathcal{N}(0,P)$. For DL messages, we first divide $W^{[\sf d]}_{ij}$, $i\in[1:K]$ and $j\in[1:N]$, into submessages $\{W^{[\sf d]}_{ij,k}\}_{k\in[1:M]}$. Then, we further divide $W^{[\sf d]}_{ij,k}$ into $T_{\sf d}^{KM(KN+KM)}$ submessages $\{W^{[\sf d](\mathbf{s}^{[\sf d]})}_{ij,k}\}_{\mathbf{s}^{[\sf d]}\in \mathcal{S}_{T_{\sf d}}^{[\sf d]}}$, where $T_{\sf d}=\left(\frac{\lambda_2}{\lambda_1M}^{\frac{1}{KM(KN+KM)}}\right)T^{\frac{N}{M}}-1$. A length-$n$ Gaussian codeword associated with $W^{[\sf d](\mathbf{s}^{[\sf d]})}_{ij,k}$ is denoted by
$\big[\begin{array}{cccc}c^{[\sf d](\mathbf{s}^{[\sf d]})}_{ij,k}[1],& c^{[\sf d](\mathbf{s}^{[\sf d]})}_{ij,k}[2], &\cdots, & c^{[\sf d](\mathbf{s}^{[\sf d]})}_{ij,k}[n]\end{array} \big]$, where its coefficients are generated i.i.d. from $\mathcal{N}(0,P)$.

Let $d=\frac{1}{\lambda_1}(T+1)^{KN(KM+KN)}$. Communication takes place over a block of $nd$ time slots. At every time slot, each of the codewords defined above is transmitted through a length-$d$ time-extended beamforming vector.
We denote the length-$d$ time-extended transmit and received signal vectors as
\begin{align}
\bar{\mathbf{x}}^{[\sf u]}_{ij}[m]&=\left[\begin{array}{cccc} x^{[\sf u]}_{ij}[(m-1)d+1], & x^{[\sf u]}_{ij}[(m-1)d+2], &\cdots, &x^{[\sf u]}_{ij}[md] \end{array}\right]^\dagger\in\mathbb{R}^{d\times 1},\nonumber\\
\bar{\mathbf{x}}^{[\sf bs]}_{i}[m]&=\left[\begin{array}{cccc} \mathbf{x}^{[\sf bs]}_{i}[(m-1)d+1], & \mathbf{x}^{[\sf bs]}_{i}[(m-1)d+2], &\cdots, &\mathbf{x}^{[\sf bs]}_{i}[md] \end{array}\right]^\dagger\in\mathbb{R}^{Md\times 1},\nonumber\\
\bar{\mathbf{y}}^{[\sf bs]}_{i}[m]&=\left[\begin{array}{cccc} \mathbf{y}^{[\sf bs]}_{i}[(m-1)d+1], & \mathbf{y}^{[\sf bs]}_{i}[(m-1)d+2], &\cdots, &\mathbf{y}^{[\sf bs]}_{i}[md] \end{array}\right]^\dagger\in\mathbb{R}^{Md\times 1},\nonumber\\
\bar{\mathbf{y}}^{[\sf d]}_{ij}[m]&=\left[\begin{array}{cccc} y^{[\sf d]}_{ij}[(m-1)d+1], & y^{[\sf d]}_{ij}[(m-1)d+2], &\cdots, &y^{[\sf d]}_{ij}[md] \end{array}\right]^\dagger\in\mathbb{R}^{d\times 1},
\end{align}
where $m\in[1:n]$, $i\in[1:K]$, and $j\in[1:N]$. Then, from \eqref{channel model1} and \eqref{channel model2},
\begin{align}\label{eq:input-output-extended}
\bar{\mathbf{y}}^{[\sf d]}_{ij}[m]&=\sum_{k=1}^{K}\mathbf{\bar{G}}_{ij,k}[m]\bar{\mathbf{x}}^{[\sf bs]}_{k}[m]+\sum_{k=1}^{K}\sum_{l=1}^{N}\mathbf{\bar{H}}_{ij,kl}[m]\bar{\mathbf{x}}^{[\sf u]}_{kl}[m]+\mathbf{\bar{z}}_{ij}^{[\sf d]}[m], \nonumber \\
\bar{\mathbf{y}}^{[\sf bs]}_{i}[m]&=\sum_{j=1}^{K}\sum_{k=1}^{N}\mathbf{\bar{F}}_{i,jk}[m]\bar{\mathbf{x}}^{[\sf u]}_{jk}[m]+\sum_{j=1}^K\mathbf{\bar{B}}_{ij}[m]\bar{\mathbf{x}}^{[\sf bs]}_{j}[m]+\mathbf{\bar{z}}^{[\sf bs]}_i[m], 
\end{align}
where
\begin{align}
\mathbf{\bar{G}}_{ij,k}[m]&=\text{diag}(\mathbf{g}_{ij,k}[(m-1)d+1],\cdots,\mathbf{g}_{ij,k}[md])\in\mathbb{R}^{d\times Md},\nonumber\\
\mathbf{\bar{F}}_{i,jk}[m]&=\text{diag}(\mathbf{f}_{i,jk}[(m-1)d+1],\cdots,\mathbf{f}_{i,jk}[md])\in\mathbb{R}^{Md\times d},\nonumber\\
\mathbf{\bar{H}}_{ij,kl}[m]&=\text{diag}(h_{ij,kl}[(m-1)d+1],\cdots,h_{ij,kl}[md])\in\mathbb{R}^{d\times d},\nonumber\\
\mathbf{\bar{B}}_{ij}[m]&=\text{diag}(\mathbf{B}_{i,j}[(m-1)d+1],\cdots,\mathbf{B}_{i,j}[md])\in\mathbb{R}^{Md\times Md}
\end{align}
and
\begin{align}
\mathbf{\bar{z}}^{[\sf bs]}_i[m]&=\left[\begin{array}{cccc} \mathbf{z}^{[\sf bs]}_{i}[(m-1)d+1], & \mathbf{z}^{[\sf bs]}_{i}[(m-1)d+2], &\cdots, &\mathbf{z}^{[\sf bs]}_{i}[md] \end{array}\right]\in\mathbb{R}^{Md\times 1},\nonumber\\
\mathbf{\bar{z}}^{[\sf d]}_{ij}[m]&=\left[\begin{array}{cccc}z^{[\sf d]}_{ij}[(m-1)d+1], & z^{[\sf d]}_{ij}[(m-1)d+2], &\cdots, &z^{[\sf d]}_{ij}[md] \end{array}\right]\in\mathbb{R}^{d\times 1}.
\end{align}

We further define $\mathbf{\bar{B}}_{ij,kl}[m]$, $\mathbf{\bar{B}}_{i,kl}[m]$, $\mathbf{\bar{G}}_{ij,kl}[m]$, and $\mathbf{\bar{F}}_{ij,kl}[m]$ as
\begin{align}
\mathbf{\bar{B}}_{ij,kl}[m]&=\text{diag}(b_{ij,kl}[(m-1)d+1],\cdots,b_{ij,kl}[md])\in\mathbb{R}^{d\times d},\nonumber\\
\mathbf{\bar{B}}_{i,kl}[m]&=\text{diag}(\mathbf{b}_{i,kl}[(m-1)d+1],\cdots,\mathbf{b}_{i,kl}[md])\in\mathbb{R}^{Md\times d},\nonumber\\
\mathbf{\bar{G}}_{ij,kl}[m]&=\text{diag}({g}_{ij,kl}[(m-1)d+1],\cdots,{g}_{ij,kl}[md])\in\mathbb{R}^{d\times d},\nonumber\\
\mathbf{\bar{F}}_{ij,kl}[m]&=\text{diag}(f_{ij,kl}[(m-1)d+1],\cdots,f_{ij,kl}[md])\in\mathbb{R}^{d\times d}.
\end{align}

\subsubsection{Transmit Beamforming for UL IA} \label{subsec:UL_trans} First, we will explain the transmit beamforming strategy of each UL user, which is designed for aligning UL inter-cell and user-to-user interferences. To this end, we adapt a recently-developed asymptotic signal space IA framework in~\cite{Ilan13,JeonSuh:14}.

For $m\in [1:n]$ and $\mathbf{s}^{[\sf u]}\in \mathcal{S}_T^{[\sf u]}$, $c^{[\sf u](\mathbf{s}^{[\sf u]})}_{ij}[m]$ is transmitted by a length-$d$ time-extended beamforming vector $\mathbf{\bar{v}}_{ij}^{[\sf u](\mathbf{s}^{[\sf u]})}[m]$ as follows:
\begin{align}\label{eq:Tx_BF1_UL}
\bar{\mathbf{x}}^{[\sf u]}_{ij}[m]=\gamma^{[\sf u]}_{ij}\sum_{\mathbf{s}^{[\sf u]}\in \mathcal{S}_T^{[\sf u]}}\mathbf{\bar{v}}_{ij}^{[\sf u](\mathbf{s}^{[\sf u]})}[m]c^{[\sf u](\mathbf{s}^{[\sf u]})}_{ij}[m],
\end{align}
where $\gamma^{[\sf u]}_{ij}$ is chosen to satisfy the average power constraint $P$. Since the construction of $\mathbf{\bar{v}}_{ij}^{[\sf u](\mathbf{s}^{[\sf u]})}[m]$ is identical for all $m\in[1:n]$, we assume $m=1$ and omit the index $m$ from now on.

For $\mathbf{s}^{[\sf u]}=\left[\begin{array}{ccccc}s_{1,1}^{[\sf u]},& s_{1,2}^{[\sf u]},&\cdots, & s_{(KN+KM),KN}^{[\sf u]}\end{array}\right],$ we define
\begin{align} \label{eq:Tx_BF1_overall_UL}
v^{[\sf u]{(\mathbf{s}^{[\sf u]})}}[t]=\prod_{1\leq p_1\leq (KN+KM), ~1\leq p_2\leq KN}{\alpha_{p_1,p_2}[t]}^{s_{p_1,p_2}^{[\sf u]}}
\end{align}
for $t\in [1:d]$ and $\mathbf{\bar{v}}^{[\sf u]{(\mathbf{s}^{[\sf u]})}}=\left[\begin{array}{ccccc}v^{[\sf u]{(\mathbf{s}^{[\sf u]})}}[1],& v^{[\sf u]{(\mathbf{s}^{[\sf u]})}}[2] ,&\cdots, & v^{[\sf u]{(\mathbf{s}^{[\sf u]})}}[d]\end{array}\right]^\dagger$, where
\begin{align}\label{eq:Tx_BF1_interference_UL}
&\alpha_{p_1,p_2}[t]=\left\{\begin{array}{lll} h_{ij,kl}[t] \quad~~ \text{if $p_1\leq KN$,}\\
f_{qr,kl}[t] \quad~~ \text{if $p_1> KN$ and $i\neq k$,}\\
1  \quad\quad\quad~~~ \text{otherwise,}
\end{array} \right.
\end{align}
and $i=\lceil\frac{p_1}{N}\rceil$, $j=p_1-N(\lceil\frac{p_1}{N}\rceil-1)$, $k=\lceil\frac{p_2}{N}\rceil$, $l=p_2-N(\lceil\frac{p_2}{N}\rceil-1)$, $q=\left\lceil\frac{p_1-KN}{M}\right\rceil$, $r=p_1-M(\left\lceil\frac{p_1-KN}{M}\right\rceil-1)$.
Then, we set $\mathbf{\bar{v}}_{ij}^{[\sf u]{(\mathbf{s}^{[\sf u]})}}=\mathbf{\bar{v}}^{[\sf u]{(\mathbf{s}^{[\sf u]})}}$
for all $i\in[1:K]$, $j\in[1:N]$, and $\mathbf{s}^{[\sf u]}\in\mathcal{S}_T^{[\sf u]}$.

\subsubsection{Transmit Beamforming for DL IA} We will now explain the transmit beamforming strategy of each BS, which is designed for aligning DL intra-cell, DL inter-cell, and BS-to-BS interferences. Here, we adapt an IA scheme in~\cite{Viveck2:09}, originally proposed for $X$-networks by treating each transmit and receive antenna of BSs as a virtual user.

Let $\bar{\mathbf{x}}^{[\sf bs]}_{i,k}[m]$ denote the length-$d$ time-extended transmit signal vector for the $k$th transmit antenna of BS $i$. For $m\in [1:n]$, $k\in[1:M]$, and $\mathbf{s}^{[\sf d]}\in \mathcal{S}_{T_{\sf d}^{[\sf d]}}$, $c^{[\sf d](\mathbf{s}^{[\sf d]})}_{ij,k}[m]$ is transmitted by a length-$d$ time-extended beamforming vector $\mathbf{\bar{v}}_{ij,k}^{[\sf bs](\mathbf{s}^{[\sf d]})}[m]$ as follows:
\begin{align}\label{eq:Tx_BF1_DL}
\bar{\mathbf{x}}^{[\sf bs]}_{i,k}[m]=\gamma^{[\sf d]}_{i,k}\sum_{j=1}^{N}\sum_{\mathbf{s}^{[\sf d]}\in \mathcal{S}_{T_{\sf d}}^{[\sf d]}}\mathbf{\bar{v}}_{ij,k}^{[\sf bs](\mathbf{s}^{[\sf d]})}[m]c^{[\sf d](\mathbf{s}^{[\sf d]})}_{ij,k}[m],
\end{align}
where $\gamma^{[\sf d]}_{i,k}$ is chosen to satisfy the average power constraint $P/M$, which is set to be equal for all transmit antennas. Similar to the UL case, since the construction of $\mathbf{\bar{v}}_{ij,k}^{[\sf bs](\mathbf{s}^{[\sf d]})}[m]$ is identical for all $m\in[1:n]$, we assume $m=1$ and omit the index $m$ in the sequel.

For $\mathbf{s}^{[\sf d]}=\left[\begin{array}{ccccc}s_{1,1}^{[\sf d]},& s_{1,2}^{[\sf d]},&\cdots, & s_{(KN+KM),KM}^{[\sf d]}\end{array}\right],$ we define
\begin{align} \label{eq:Tx_BF1_overall_DL}
v_{j}^{[\sf bs]{(\mathbf{s}^{[\sf d]})}}[t]=\prod_{1\leq p_1 \leq (KN+KM), ~1\leq p_2\leq KM}\beta_{j,p_1,p_2}[t]^{s^{[\sf d]}_{p_1,p_2}},
\end{align}
for $t\in [1:d]$ and $\mathbf{\bar{v}}_{j}^{[\sf bs]{(\mathbf{s}^{[\sf d]})}}=\left[\begin{array}{ccccc}v_{j}^{[\sf bs]{(\mathbf{s}^{[\sf d]})}}[1],& v_{j}^{[\sf bs]{(\mathbf{s}^{[\sf d]})}}[2] ,&\cdots, & v_{j}^{[\sf bs]{(\mathbf{s}^{[\sf d]})}}[d]\end{array}\right]^\dagger$, where
\begin{align}\label{eq:Tx_BF1_interference_DL}
\beta_{j,p_1,p_2}[t]=\left\{\begin{array}{lll} &g_{ik,lq}[t] \quad~~\text{if $p_1\leq KN$ and $i\neq l$,}\\
&g_{ik,lq}[t] \quad~~\text{if $p_1\leq KN$ and $j\neq k$,}\\
&b_{rs,lq}[t]  \quad~~ \text{if $p_1>KN$ and $r\neq l$, }\\
&1  \quad\quad\quad~~~ \text{otherwise, }
\end{array} \right.
\end{align}
and $i=\lceil\frac{p_1}{N}\rceil$, $k=p_1-N(\lceil\frac{p_1}{N}\rceil-1)$, $l=\lceil\frac{p_2}{M}\rceil$, $q=p_2-M(\lceil\frac{p_2}{M}\rceil-1)$, $r=\left\lceil\frac{p_1-KN}{M}\right\rceil$, $s=p_1-M(\left\lceil\frac{p_1-KN}{M}\right\rceil-1)$. Then, we set $\mathbf{\bar{v}}_{ij,k}^{[\sf bs]{(\mathbf{s}^{[\sf d]})}}=\mathbf{\bar{v}}_{j}^{[\sf bs]{(\mathbf{s}^{[\sf d]})}}$
 for all $i\in[1:K]$, $j\in[1:N]$, $k\in[1:M]$, and $\mathbf{s}^{[\sf d]}\in\mathcal{S}_{T_{\sf d}^{[\sf d]}}$.\\

\subsubsection{Decoding at BSs}
From the proposed beamforming strategy, the length-$d$ time extended received signal vector of BS $i$ is given by
\begin{align}
\bar{\mathbf{y}}^{[\sf bs]}_{i}&=\underbrace{\sum_{k=1}^{N}\gamma^{[\sf u]}_{ik}\sum_{\mathbf{s}^{[\sf u]}\in \mathcal{S}_T^{[\sf u]}}\mathbf{\bar{F}}_{i,ik}\mathbf{\bar{v}}_{ik}^{[\sf u](\mathbf{s}^{[\sf u]})}c^{[\sf u](\mathbf{s}^{[\sf u]})}_{ik}}_{\text{Desired signals}}\nonumber\\
&+\underbrace{\sum_{j=1,j\neq i }^{K}\gamma^{[\sf u]}_{jk}\sum_{k=1}^{N}\sum_{\mathbf{s}^{[\sf u]}\in \mathcal{S}_T^{[\sf u]}}\mathbf{\bar{F}}_{i,jk}\mathbf{\bar{v}}_{jk}^{[\sf u](\mathbf{s}^{[\sf u]})}c^{[\sf u](\mathbf{s}^{[\sf u]})}_{jk}}_{\text{UL inter-cell interferences}}\nonumber\\
&+\underbrace{\sum_{j=1}^K\sum_{l=1}^M \gamma^{[\sf d]}_{j,l}\sum_{k=1}^{N}\sum_{\mathbf{s}^{[\sf d]}\in \mathcal{S}_{T_{\sf d}}^{[\sf d]}}\mathbf{\bar{B}}_{i,jl}\mathbf{\bar{v}}_{jk,l}^{[\sf bs](\mathbf{s}^{[\sf d]})}c^{[\sf d](\mathbf{s}^{[\sf d]})}_{jk,l}}_{\text{BS-to-BS interferences}}+\mathbf{\bar{z}}^{[\sf bs]}_i.
\end{align}

Now we will examine the dimensions of the desired signals and interference signals occupied at each BS. First, we investigate the dimension of UL inter-cell interferences. Since the cardinality of $\left\{\mathbf{\bar{v}}_{jk}^{[\sf u](\mathbf{s}^{[\sf u]})}\right\}_{j\in[1:K]\setminus i,k\in[1:N],\mathbf{s^{[\sf u]}}\in \mathcal{S}_T^{[\sf u]}}$ is given by $(K-1)N T^{KN(KN+KM)}$, 
\begin{align}
\text{span}\left(\left\{\bar{\mathbf{F}}_{i,jk}\mathbf{\bar{v}}_{jk}^{[\sf u](\mathbf{s}^{[\sf u]})}\right\}_{j\in[1:K]\setminus i,k\in[1:N],\mathbf{s}^{[\sf u]}\in \mathcal{S}_T^{[\sf u]}}\right)
\end{align}
occupies at most $(K-1)N T^{KN(KN+KM)}$ dimensional subspace. In addition, from the UL beamforming construction
\eqref{eq:Tx_BF1_overall_UL}, for given $i\in[1:K]$ and $l\in[1:M]$,
\begin{align}
\bar{\mathbf{F}}_{il,jk}\mathbf{\bar{v}}_{jk}^{[\sf u](\mathbf{s}^{[\sf u]})}\in \left\{\mathbf{\bar{v}}_{jk}^{[\sf u]({\mathbf{s}^{[\sf u]}}')}\right\}_{{\mathbf{s}^{[\sf u]}}'\in \mathcal{S}^{[\sf u]}_{T+1}}
\end{align}
for all $j\in [1:K]\setminus i$, $k\in[1:N]$, and $\mathbf{s}^{[\sf u]}\in \mathcal{S}_T^{[\sf u]}$, showing that 
\begin{align}
\text{span}\left(\left\{\bar{\mathbf{F}}_{il,jk}\mathbf{\bar{v}}_{jk}^{[\sf u](\mathbf{s}^{[\sf u]})}\right\}_{j\in[1:K]\setminus i,k\in[1:N],\mathbf{s}^{[\sf u]}\in \mathcal{S}_T^{[\sf u]}}\right)
\end{align}
occupies at most $(T+1)^{KN(KN+KM)}$ dimensional subspace due to the fact that the cardinality of $\mathcal{S}_{T+1}^{[\sf u]}$ is given by $(T+1)^{KN(KN+KM)}$, and thus 
\begin{align}
\text{span}\left(\left\{\bar{\mathbf{F}}_{i,jk}\mathbf{\bar{v}}_{jk}^{[\sf u](\mathbf{s}^{[\sf u]})}\right\}_{j\in[1:K]\setminus i,k\in[1:N],\mathbf{s}^{[\sf u]}\in \mathcal{S}_T^{[\sf u]}}\right)
\end{align}
occupies at most $M(T+1)^{KN(KN+KM)}$ dimensional subspace. Consequently, the dimension of UL inter-cell interferences is at most $\min\{(K-1)N T^{KN(KN+KM)}, M(T+1)^{KN(KN+KM)}\}\leq \min\{(K-1)N,M\} (T+1)^{KN(KN+KM)}$.

Now we investigate the dimension of BS-to-BS interferences. From the DL beamforming construction \eqref{eq:Tx_BF1_overall_DL}, for given $i\in[1:K]$, $k\in[1:M]$, and $p\in[1:N]$,
\begin{align}
\mathbf{\bar{B}}_{ik,jl}\mathbf{\bar{v}}_{jp,l}^{[\sf bs](\mathbf{s}^{[\sf d]})}\in\left\{\mathbf{\bar{v}}_{jp,l}^{[\sf bs](\mathbf{s^{[\sf d]}}')}\right\}_{\mathbf{s^{[\sf d]}}'\in \mathcal{S}^{[\sf d]}_{T_{\sf d}+1}}
\end{align}
for all $j\in[1:K]$, $l\in[1:M]$, and $\mathbf{s}\in \mathcal{S}_{T_{\sf d}}^{[\sf d]}$, showing that 
\begin{align}
\text{span}\left(\left\{\mathbf{\bar{B}}_{ik,jl}\mathbf{\bar{v}}_{jp,l}^{[\sf bs](\mathbf{s}^{[\sf d]})}\right\}_{ j\in[1:K], l\in[1:M], \mathbf{s}^{[\sf d]}\in \mathcal{S}_{T_{\sf d}}^{[\sf d]} }\right)
\end{align}
occupies at most $(T_{\sf d}+1)^{KM(KN+KM)}$ dimensional subspace due to the fact that the cardinality of $\mathcal{S}_{T_{\sf d}+1}^{[\sf d]}$ is given by $(T_{\sf d}+1)^{KM(KN+KM)}$. Therefore, for given $i\in[1:K]$ and $k\in[1:M]$, 
\begin{align}
\text{span}\left(\left\{\mathbf{\bar{B}}_{ik,jl}\mathbf{\bar{v}}_{jp,l}^{[\sf bs](\mathbf{s}^{[\sf d]})}\right\}_{j\in[1:K], l\in[1:M], p\in[1:N], \mathbf{s}^{[\sf d]}\in \mathcal{S}_{T_{\sf d}}^{[\sf d]} }\right)
\end{align}
 occupies at most $N(T_{\sf d}+1)^{KM(KN+KM)}$ dimensional subspace and, consequently,
\begin{align}
\text{span}\left(\left\{\mathbf{\bar{B}}_{i,jl}\mathbf{\bar{v}}_{jp,l}^{[\sf bs](\mathbf{s}^{[\sf d]})}\right\}_{j\in[1:K], l\in[1:M], p\in[1:N], \mathbf{s}^{[\sf d]}\in \mathcal{S}_{T_{\sf d}}^{[\sf d]} }\right)
\end{align}
occupies at most $MN(T_{\sf d}+1)^{KM(KN+KM)}$ dimensional subspace. In summary, the overall dimension of UL inter-cell and BS-to-BS interferences at BS $i$ is at most 
\begin{align}
\min\{(K-1)N,M\} (T+1)^{KN(KN+KM)}+MN(T_{\sf d}+1)^{KM(KN+KM)}.
\end{align}

Finally, we examine the dimension and linear independence property of the desired UL signals. Recall that  $\left\{\mathbf{\bar{v}}_{ij}^{[\sf u](\mathbf{s}^{[\sf u]})}\right\}_{\mathbf{s}^{[\sf u]}\in \mathcal{S}_T^{[\sf u]}}$ is a set of linearly independent vectors whose cardinality is given by $T^{KN(KN+KM)}$, which is constructed independent of $\{\mathbf{\bar{F}}_{i,ij}\}_{j\in[1:N]}$ (see \eqref{eq:Tx_BF1_overall_UL} and \eqref{eq:Tx_BF1_interference_UL}). Therefore,  $\left\{\mathbf{\bar{F}}_{i,ij}\mathbf{\bar{v}}_{ij}^{[\sf u](\mathbf{s}^{[\sf u]})}\right\}_{j\in[1:N], \mathbf{s}^{[\sf u]}\in \mathcal{S}_T^{[\sf u]}}$, the desired UL signals for BS $i$, is a set of linearly independent vectors whose cardinality is given by $NT^{KN(KN+KM)}$. Furthermore, since channels are generic,
each of $\left\{\mathbf{\bar{F}}_{i,ij}\mathbf{\bar{v}}_{ij}^{[\sf u](\mathbf{s}^{[\sf u]})}\right\}_{j\in[1:N], \mathbf{s}^{[\sf u]}\in \mathcal{S}_T^{[\sf u]}}$ is linearly independent of UL inter-cell and BS-to-Bs interferences if
\begin{align}
&NT^{KN(KN+KM)}+\min\{(K-1)N,M\} (T+1)^{KN(KN+KM)}+MN(T_{\sf d}+1)^{KM(KN+KM)}\nonumber\\
&=NT^{KN(KN+KM)}+\min\{(K-1)N,M\} (T+1)^{KN(KN+KM)}+N\frac{\lambda_2}{\lambda_1}T^{KN(KN+KM)}\nonumber\\
&\leq \frac{M}{\lambda_1}(T+1)^{KN(KM+KN)},
\end{align}
which is satisfied from the assumption $\lambda_1(N+\min\{M,K(N-1)\}+\lambda_2N\leq M$ in \eqref{eq:Dof1}. Therefore, BS $i$ can recover its desired UL messages by nulling out all UL inter-cell and BS-to-BS interferences.
\subsubsection{Decoding at DL users} From the proposed beamforming strategy, the length-$d$ time extended received signal vector of DL user $(i,j)$ is given by
\begin{align}
\bar{\mathbf{y}}^{[\sf d]}_{ij}&=\underbrace{\sum_{l=1}^{M}\gamma^{[\sf d]}_{i,l}\sum_{\mathbf{s}^{[\sf d]}\in \mathcal{S}_{T_{\sf d}}^{[\sf d]}}\mathbf{\bar{G}}_{ij,il}\mathbf{\bar{v}}_{ij,l}^{[\sf bs](\mathbf{s}^{[\sf d]})}c^{[\sf d](\mathbf{s}^{[\sf d]})}_{ij,l}}_{\text{Desired signals}}\nonumber\\
&+\underbrace{\sum_{l=1}^{M}\gamma^{[\sf d]}_{i,l}\sum_{p=1, p\neq j}^{N}\sum_{\mathbf{s}^{[\sf d]}\in \mathcal{S}_{T_{\sf d}}^{[\sf d]}}\mathbf{\bar{G}}_{ij,il}\mathbf{\bar{v}}_{ip,l}^{[\sf bs](\mathbf{s}^{[\sf d]})}c^{[\sf d](\mathbf{s}^{[\sf d]})}_{ip,l}}_{\text{DL intra-cell interferences}}\nonumber\\
&+\underbrace{\sum_{k=1,k\neq i}^K\sum_{l=1}^{M}\gamma^{[\sf d]}_{i,l}\sum_{p=1}^{N}\sum_{\mathbf{s}^{[\sf d]}\in \mathcal{S}_{T_{\sf d}}^{[\sf d]}}\mathbf{\bar{G}}_{ij,kl}[m]\mathbf{\bar{v}}_{kp,l}^{[\sf bs](\mathbf{s}^{[\sf d]})}c^{[\sf d](\mathbf{s}^{[\sf d]})}_{kp,l}}_{\text{DL inter-cell interferences}}\nonumber\\
&+\underbrace{\sum_{k=1}^{K}\sum_{l=1}^{N}\gamma^{[\sf u]}_{kl}\sum_{\mathbf{s}^{[\sf u]}\in \mathcal{S}_T^{[\sf u]}}\mathbf{\bar{H}}_{ij,kl}\mathbf{\bar{v}}_{kl}^{[\sf u](\mathbf{s}^{[\sf u]})}c^{[\sf u](\mathbf{s}^{[\sf u]})}_{kl}}_{\text{User-to-user interferences}}+\mathbf{\bar{z}}_{ij}^{[\sf d]}.
\end{align}

Now we will examine the dimension of the desired signals and interference signals occupied at each DL user. First, we investigate the overall dimension of DL intra-cell and inter-cell interferences. From the DL beamforming \eqref{eq:Tx_BF1_overall_DL}, for given $i\in[1:K]$, $j\in[1:N]$, and $p\in[1:N]$,
\begin{align}
\mathbf{\bar{G}}_{ij,kl}\mathbf{\bar{v}}_{kp,l}^{[\sf bs](\mathbf{s}^{[\sf d]})}\in\left\{\mathbf{\bar{v}}_{kp,l}^{[\sf bs]({\mathbf{s}^{[\sf d]}}')}\right\}_{{\mathbf{s}^{[\sf d]}}'\in \mathcal{S}^{[\sf d]}_{T_{\sf d}+1}}
\end{align}
for all $k\in[1:K]$, $l\in[1:M]$, and $\mathbf{s}^{[\sf d]}\in \mathcal{S}_{T_{\sf d}}^{[\sf d]}$ except for the case where $p=j$ and $k=i$ showing that
\begin{align}
\text{span}\left(\left\{\mathbf{\bar{G}}_{ij,kl}\mathbf{\bar{v}}_{kp,l}^{[\sf bs](\mathbf{s}^{[\sf d]})}\right\}_{ k\in[1:K], l\in[1:M], \mathbf{s}^{[\sf d]}\in \mathcal{S}_{T_{\sf d}}^{[\sf d]}~ \text{except $p=j$ and $k=i$} }\right)
\end{align}
occupies $(T_{\sf d}+1)^{KM(KN+KM)}$ dimensional subspace due to the fact that the cardinality of $\mathcal{S}_{T_{\sf d}+1}^{[\sf d]}$ is given by $(T_{\sf d}+1)^{KM(KN+KM)}$. Therefore,
\begin{align}
\text{span}\left(\left\{\mathbf{\bar{G}}_{ij,kl}\mathbf{\bar{v}}_{kp,l}^{[\sf bs](\mathbf{s}^{[\sf d]})}\right\}_{p\in[1:N], k\in[1:K], l\in[1:M], \mathbf{s}^{[\sf d]}\in \mathcal{S}_{T_{\sf d}}^{[\sf d]}~ \text{except  $j=p$ and $k=i$} }\right)
\end{align}
occupies at most $N(T_{\sf d}+1)^{KM(KN+KM)}$ dimensional subspace, which means that the overall dimension of DL intra-cell and inter-cell interferences is at most $N(T_{\sf d}+1)^{KM(KN+KM)}$.

Now we investigate the dimension of user-to-user interferences. From the UL beamforming \eqref{eq:Tx_BF1_overall_UL}, for given $i\in[1:K]$ and $j\in[1:N]$,
\begin{align}
\bar{\mathbf{H}}_{ij,kl}\mathbf{\bar{v}}_{kl}^{[\sf u](\mathbf{s}^{[\sf u]})}\in \left\{\mathbf{\bar{v}}_{kl}^{[\sf u]({\mathbf{s}^{[\sf u]}}')}\right\}_{{\mathbf{s}^{[\sf u]}}'\in \mathcal{S}^{[\sf u]}_{T+1}}
\end{align} for all $k\in[1:K]$, $l\in[1:N]$, and $\mathbf{s}\in \mathcal{S}_T^{[\sf u]}$, showing that
\begin{align}
\text{span}\left(\left\{\bar{\mathbf{H}}_{ij,kl}\mathbf{\bar{v}}_{kl}^{[\sf u](\mathbf{s}^{[\sf u]})}[m]\right\}_{k\in[1:K], l\in[1:N],\mathbf{s}^{[\sf u]}\in \mathcal{S}_T^{[\sf u]}}\right)
\end{align} 
occupies at most $(T+1)^{KN(KN+KM)}$ dimensional subspace due to the fact that the cardinality of $\mathcal{S}_{T+1}^{[\sf u]}$ is given by $(T+1)^{KN(KN+KM)}$. In summary, the overall dimension of DL intra-cell, DL inter-cell, and user-to-user interferences is at most $N(T_{\sf d}+1)^{KM(KN+KM)}+(T+1)^{KN(KN+KM)}$.

Finally, we examine the dimension and linear independence property of the desired DL signals. Recall that  $\left\{\mathbf{\bar{v}}_{ij,k}^{[\sf bs](\mathbf{s}^{[\sf d]})}\right\}_{\mathbf{s}^{[\sf d]}\in \mathcal{S}_{T_{\sf d}}^{[\sf d]}}$ is a set of linearly independent vectors whose cardinality is given by $T_{\sf d}^{KM(KN+KM)}$, which is constructed independent of $\{\mathbf{\bar{G}}_{ij,ik}\}_{k\in[1:M]}$ (see \eqref{eq:Tx_BF1_overall_DL} and \eqref{eq:Tx_BF1_interference_DL}). Therefore,  $\left\{\mathbf{\bar{G}}_{ij,ik}\mathbf{\bar{v}}_{ij,k}^{[\sf bs](\mathbf{s}^{[\sf d]})}\right\}_{k\in[1:M], \mathbf{s}^{[\sf d]}\in \mathcal{S}_{T_{\sf d}}^{[\sf d]}}$, the desired DL signals for DL user $(i,j)$, is a set of linearly independent vectors whose cardinality is given by $MT_{\sf d}^{KM(KN+KM)}$. Furthermore, since channels are generic, each of $\left\{\mathbf{\bar{G}}_{ij,ik}\mathbf{\bar{v}}_{ij,k}^{[\sf bs](\mathbf{s}^{[\sf d]})}\right\}_{k\in[1:M], \mathbf{s}^{[\sf d]}\in \mathcal{S}_{T_{\sf d}}^{[\sf d]}}$ is linearly independent of DL intra-cell, DL inter-cell, and user-to-user interferences if
\begin{align}
&MT_{\sf d}^{KM(KN+KM)}+N(T_{\sf d}+1)^{KM(KN+KM)}+(T+1)^{KN(KN+KM)}\nonumber\\
&\leq (M+N)(T_{\sf d}+1)^{KM(KN+KM)}+(T+1)^{KN(KN+KM)}\nonumber\\
&=\frac{\lambda_2}{\lambda_1}\left(1+\frac{N}{M}\right)T^{KN(KN+KM)}+(T+1)^{KN(KN+KM)}\nonumber\\
&\leq \frac{1}{\lambda_1}(T+1)^{KN(KM+KN)},
\end{align}
which is satisfied from the assumption $\lambda_1+\lambda_2\left(1+\frac{N}{M}\right)\leq 1$ in \eqref{eq:Dof1}. Therefore, DL user $(i,j)$ can recover its desired DL message.

\subsubsection{Achievable Sum DoF} Since total $KN(T^{KN(KN+KM)}+MT_{\sf d}^{KM(KN+KM)})$ submessages are sent during $nd=\frac{n}{\lambda_1}(T+1)^{KN(KM+KN)}$ time slots via a length-$n$ codeword each, the sum DoF of
\begin{align} \label{eq:Dof1_converse}
&\frac{KNT^{KN(KN+KM)}+KNMT_{\sf d}^{KM(KN+KM)}}{\frac{1}{\lambda_1}(T+1)^{KN(KM+KN)}} \nonumber \\
&=\frac{KNT^{KN(KN+KM)}+KNM\left(\left(\frac{\lambda_2}{\lambda_1M}^{\frac{1}{KM(KN+KM)}}\right)T^{\frac{N}{M}}-1\right)^{KM(KN+KM)}}{\frac{1}{\lambda_1}(T+1)^{KN(KM+KN)}}
\end{align}
is achievable under the constraints in \eqref{eq:Dof1}. Finally, since \eqref{eq:Dof1_converse} converges to $KN(\lambda_1+\lambda_2)$ as $T$ increases, the sum DoF in \eqref{eq:Dof1} is achievable.

\subsection{Scheme 2: UL IA and DL Interference Nulling} \label{subsec:scheme2}
The second scheme applies the same UL IA at each UL user as in the first scheme for aligning UL inter-cell and user-to-user interferences. On the other hand, for DL transmission, beamforming vectors of BSs are now designed for nulling out DL intra-cell, DL inter-cell, and BS-to-BS interferences. By employing this strategy, for $M<KN$, we will show that the following  sum DoF is achievable:
\begin{align}\label{eq:Dof2}
d_{\Sigma,2}=\max_{\substack{
\lambda_1,\lambda_2\in(0,1]\\
            \lambda_1(N+\min\{M,K(N-1)\}\leq M\\
           \lambda_1(K-1)N+\lambda_2KN\leq M\\
            \lambda_1+\lambda_2\leq 1\\
           }}KN(\lambda_1+\lambda_2).
            \end{align}

As the same manner in Section \ref{subsec:scheme1}, assume that $\lambda_1$ and $\lambda_2$ satisfy the four constraints in \eqref{eq:Dof2} from now on.
Communication takes place over a block of $nd$ slots where $d=\frac{1}{\lambda_1}(T+1)^{KN(KM+KN)}$ as in the first scheme. In addition, as aforementioned, transmit beamforming at each UL user is the same as in Section \ref{subsec:UL_trans}. On the other hand, for DL transmission, $W_{ij}^{[\sf d]}$  is now divided into $\frac{\lambda_2}{\lambda_1}T^{KN(KM+KN)}$ submessages  $\{W^{{[\sf d]}(a)}_{ij}\}_{a\in[1:\frac{\lambda_2}{\lambda_1}T^{KN(KM+KN)}]}$. A length-$n$ Gaussian codeword associated with $W^{{[\sf d]}(a)}_{ij}$ is denoted by
$[\begin{array}{cccc}c^{{[\sf d]}(a)}_{ij}[1], & c^{{[\sf d]}(a)}_{ij}[2], &\cdots, & c^{{[\sf d]}(a)}_{ij}[n]\end{array} ],$ where its coefficients are generated i.i.d. from $\mathcal{N}(0,P)$. For $m\in[1:n]$ and $a\in[1:\frac{\lambda_2}{\lambda_1}T^{KN(KM+KN)}],$ $c^{[{\sf d}](a)}_{ij}[m]$ is transmitted via a length-$d$ time-extended beamforming vector $\mathbf{\bar{v}}_{ij}^{[{\sf bs}](a)}[m]\in\mathbb{R}^{Md\times 1}$ as
\begin{align}\label{eq:Tx_BF2_DL}
\bar{\mathbf{x}}^{[\sf bs]}_{i}[m]=\gamma^{[\sf d]}_{i}\sum_{j=1}^N\sum_{a=1}^{\frac{\lambda_2}{\lambda_1}T^{KN(KM+KN)}}\mathbf{\bar{v}}_{ij}^{[{\sf bs}](a)}[m]c^{[{\sf d}](a)}_{ij}[m],
\end{align}
where $\gamma^{[\sf d]}_{i}$ is chosen to satisfy the average power constraint $P$.

\subsubsection{Transmit beamforming for DL interference nulling}
From \eqref{eq:Tx_BF1_UL} and \eqref{eq:Tx_BF2_DL},  we have
\begin{align} \label{eq:input-output_BS_Scheme2}
\bar{\mathbf{y}}^{[\sf bs]}_{i}&=\underbrace{\sum_{k=1}^{N}\gamma^{[\sf u]}_{ik}\sum_{\mathbf{s}^{[\sf u]}\in \mathcal{S}_T^{[\sf u]}}\mathbf{\bar{F}}_{i,ik}\mathbf{\bar{v}}_{ik}^{[\sf u](\mathbf{s}^{[\sf u]})}c^{[\sf u](\mathbf{s}^{[\sf u]})}_{ik}}_{\text{Desired signals}} \nonumber \\
&+\underbrace{\sum_{j=1,i\neq j}^{K}\gamma^{[\sf u]}_{jk}\sum_{k=1}^{N}\sum_{\mathbf{s}\in \mathcal{S}_T^{[\sf u]}}\mathbf{\bar{F}}_{i,jk}\mathbf{\bar{v}}_{jk}^{[\sf u](\mathbf{s}^{[\sf u]})}c^{[\sf u](\mathbf{s}^{[\sf u]})}_{jk}}_{\text{UL inter-cell interferences}} \nonumber \\
&+\underbrace{\sum_{j=1}^K\gamma^{[\sf d]}_{j}\sum_{k=1}^N\sum_{a=1}^{\frac{\lambda_2}{\lambda_1}T^{KN(KM+KN)}}\mathbf{\bar{B}}_{ij}\mathbf{\bar{v}}_{jk}^{[{\sf bs}](a)}c^{[{\sf d}](a)}_{jk}}_{\text{BS-to-BS interferences}}+\mathbf{\bar{z}}^{[\sf bs]}_i
\end{align} 
and
\begin{align} \label{eq:input-output_user_Scheme2}
\bar{\mathbf{y}}^{[\sf d]}_{ij}
&=\underbrace{\gamma^{[\sf d]}_{i}\sum_{a=1}^{\frac{\lambda_2}{\lambda_1}T^{KN(KM+KN)}}\mathbf{\bar{G}}_{ij,i}\mathbf{\bar{v}}_{ij}^{[{\sf bs}](a)}c^{[{\sf u}](a)}_{ij}}_{\text{Desired signals}} \nonumber\\
&+\underbrace{\gamma^{[\sf d]}_{i}\sum_{k=1, k\neq j}^N\sum_{a=1}^{\frac{\lambda_2}{\lambda_1}T^{KN(KM+KN)}}\mathbf{\bar{G}}_{ij,i}\mathbf{\bar{v}}_{ik}^{[{\sf bs}](a)}c^{[{\sf d}](a)}_{ik}}_{\text{DL intra-cell interferences}} \nonumber\\
&+\underbrace{\sum_{p=1, p\neq i}^{K}\gamma^{[\sf d]}_{p}\sum_{k=1}^N\sum_{a=1}^{\frac{\lambda_2}{\lambda_1}T^{KN(KM+KN)}}\mathbf{\bar{G}}_{ij,p}\mathbf{\bar{v}}_{pk}^{[{\sf bs}](a)}c^{[{\sf d}](a)}_{pk}}_{\text{DL inter-cell interferences}} \nonumber\\
&+\underbrace{\sum_{k=1}^{K}\sum_{l=1}^{N}\gamma^{[\sf u]}_{kl}\sum_{\mathbf{s}^{[\sf u]}\in \mathcal{S}_T^{[\sf u]}}\mathbf{\bar{H}}_{ij,kl}\mathbf{\bar{v}}_{kl}^{[\sf u](\mathbf{s}^{[\sf u]})}c^{[\sf u](\mathbf{s}^{[\sf u]})}_{kl}}_{\text{User-to-user interferences}}+\mathbf{\bar{z}}_{ij}^{[\sf d]}.
\end{align}

From \eqref{eq:input-output_BS_Scheme2}, in order to null out BS-to-BS interferences by zero-forcing at BS $i$,
\begin{align} \label{eq:condition_IN_inter-BS_Scheme2}
\mathbf{\bar{B}}_{qi}\mathbf{\bar{v}}_{ij}^{[{\sf bs}](a)}\perp \text{span}\left(\left\{\mathbf{\bar{F}}_{q,qk}\mathbf{\bar{v}}_{qk}^{[\sf u](\mathbf{s}^{[\sf u]})}\right\}_{k\in[1:N], \mathbf{s}^{[\sf u]}\in \mathcal{S}_T^{[\sf u]}}\right)
\end{align}
should be satisfied for all $q\in[1:K]\setminus i$, $j\in[1:K]$, and $a\in[1:\frac{\lambda_2}{\lambda_1}T^{KN(KM+KN)}].$ Similarly, from~\eqref{eq:input-output_user_Scheme2}, in order to null out DL inter-cell interferences by zero-forcing at BS $i$,
\begin{align} \label{eq:condition_IN_inter-cell_Scheme2}
\mathbf{\bar{G}}_{qp,i}\mathbf{\bar{v}}_{ij}^{[{\sf bs}](a)}\perp \text{span}\left(\left\{\mathbf{\bar{G}}_{qp,q}\mathbf{\bar{v}}_{qp}^{{[\sf bs]}(a')}\right\}_{a'\in[1:\frac{\lambda_2}{\lambda_1}T^{KN(KM+KN)}]}\right)
\end{align}
should be satisfied for all $q\in[1:K]\setminus i$, $j\in[1:N]$, $p\in[1:N]$, and $a\in[1:\frac{\lambda_2}{\lambda_1}T^{KN(KM+KN)}]$.

In order to null out DL intra-cell interferences, 
we first construct $\left\{\bar{\mathbf{w}}_{a', ij}\right\}_{a'\in[1:\frac{\lambda_2}{\lambda_1}T^{KN(KM+KN)}]}$ by choosing $\frac{\lambda_2}{\lambda_1}T^{KN(KM+KN)}$ basis vectors consisting of the null space of
\begin{align} \label{eq:inter-user interference}
\text{span}\left(\left\{\bar{\mathbf{H}}_{ij,kl}\mathbf{\bar{v}}_{kl}^{[\sf u](\mathbf{s}^{[\sf u]})}\right\}_{k\in[1:K], l\in[1:N],\mathbf{s}^{[\sf u]}\in \mathcal{S}_T^{[\sf u]}}\right),
\end{align}
which is used for the signal space of the desired submessages for DL user $(i,j)$.
This is possible since \eqref{eq:inter-user interference} occupies at most $(T+1)^{KN(KN+KM)}$ dimensional subspace, which implies that the null space of \eqref{eq:inter-user interference} occupies at least $(\frac{1}{\lambda_1}-1)(T+1)^{KN(KM+KN)}$, and
\begin{align}
\left(\frac{1}{\lambda_1}-1\right)(T+1)^{KN(KM+KN)}\geq \frac{\lambda_2}{\lambda_1}T^{KN(KM+KN)}
\end{align}
is satisfied from the condition $\lambda_1+\lambda_2\leq 1$ in \eqref{eq:Dof2}. Hence, in order to null out DL intra-cell interferences by zero-forcing at BS $i$,
\begin{align} \label{eq:condition_IN_intra-cell_Scheme2}
\mathbf{\bar{G}}_{ik,i}\mathbf{\bar{v}}_{ij}^{[{\sf bs}](a)}\perp \text{span}\left(\left\{\bar{\mathbf{w}}_{a', ik}\right\}_{a'\in[1:\frac{\lambda_2}{\lambda_1}T^{KN(KM+KN)}]}\right)
\end{align}
should be satisfied for all $j\in[1:N]$, $k\in[1:N]\setminus j$, and $a\in[1:\frac{\lambda_2}{\lambda_1}T^{KN(KM+KN)}].$

As a result, from \eqref{eq:condition_IN_inter-BS_Scheme2}, \eqref{eq:condition_IN_inter-cell_Scheme2}, and \eqref{eq:condition_IN_intra-cell_Scheme2}, $\mathbf{\bar{v}}_{ij}^{[{\sf bs}](a)}$ should be orthogonal with the following vector spaces:
\begin{align}\label{condition_Scheme2}
&\text{span}\left(\left\{\mathbf{\bar{B}}_{qi}^{-1}\mathbf{\bar{F}}_{q,qk}\mathbf{\bar{v}}_{qk}^{[\sf u](\mathbf{s}^{[\sf u]})}\right\}_{q\in[1:K]\setminus i, k\in[1:N],\mathbf{s}^{[\sf u]}\in \mathcal{S}_T^{[\sf u]}}\right),\nonumber \\
&\text{span}\left(\left\{\mathbf{\bar{G}}_{qp,i}^{-1}\mathbf{\bar{G}}_{qp,q}\mathbf{\bar{v}}_{qp}^{{[\sf bs]}(a')}\right\}_{ q\in[1:K]\setminus i, p\in[1:N],a'\in[1:\frac{\lambda_2}{\lambda_1}T^{KN(KM+KN)}]}\right),\nonumber \\
&\text{span}\left(\{\mathbf{\bar{G}}_{ik,i}^{-1}\bar{\mathbf{w}}_{a',ik}\}_{k\in[1:N]\setminus j,a'\in[1:\frac{\lambda_2}{\lambda_1}T^{KN(KM+KN)}]}\right).
\end{align}
Note that the total of $((K-1)N+(K-1)N\frac{\lambda_2}{\lambda_1}+(N-1)\frac{\lambda_2}{\lambda_1})T^{KN(KM+KN)}$ vectors are in \eqref{condition_Scheme2} and $\mathbf{\bar{v}}_{ij}^{[{\sf bs}](a)}$ is a vector with $Md=\frac{M}{\lambda_1}(T+1)^{KN(KM+KN)}$ elements. Hence, we can choose linearly independent $\left\{\mathbf{\bar{v}}_{ij}^{[{\sf bs}](a)}\right\}_{a\in[1:\frac{\lambda_2}{\lambda_1}T^{KN(KM+KN)}]}$ orthogonal with the vector spaces in \eqref{condition_Scheme2} if
\begin{align}
&\frac{M}{\lambda_1}(T+1)^{KN(KM+KN)}-\left((K-1)NT^{KN(KN+KM)}+(KN-1)\frac{\lambda_2}{\lambda_1}T^{KN(KM+KN)}\right)\nonumber\\
&>\frac{\lambda_2}{\lambda_1}T^{KN(KM+KN)},
\end{align}
which is satisfied from the assumption $(K-1)N\lambda_1+KN \lambda_2\leq M$ in \eqref{eq:Dof2}. Hence, we can set beamforming vectors of BSs for nulling DL intra-cell, DL inter-cell, and BS-to-BS interferences.

\subsubsection{Decoding at BSs} Since $\left\{\mathbf{\bar{v}}_{ij}^{[{\sf bs}](a)}\right\}_{i\in[1:K], j\in[1:N], a\in[1:\frac{\lambda_2}{\lambda_1}T^{KN(KM+KN)}]}$ is designed to partially null out BS-to-BS interferences, they can be removed by receive zero-forcing at each BS. 
In addition, recall that the dimension of UL inter-cell interferences is at most 
\begin{align}
\min\{(K-1)N,M\} (T+1)^{KN(KN+KM)}
\end{align}
 from the proposed UL IA. Therefore, $\left\{\mathbf{\bar{F}}_{i,ij}\mathbf{\bar{v}}_{ij}^{[\sf u](\mathbf{s}^{[\sf u]})}\right\}_{j\in[1:N], \mathbf{s}^{[\sf u]}\in \mathcal{S}_T^{[\sf u]}}$ is a set of linearly independent vectors desired for BS $i$ and also linearly independent of UL inter-cell interferences if
\begin{align}
&NT^{KN(KM+KN)}+\min\{(K-1)N,M\} (T+1)^{KN(KM+KN)}\leq \frac{M}{\lambda_1}(T+1)^{KN(KM+KN)},
\end{align}
which is satisfied from the assumption $\lambda_1(N+\min\{M,K(N-1)\}\leq M$ in \eqref{eq:Dof2}. Therefore, BS $i$ can recover its desired UL messages.

\subsubsection{Decoding at DL users}
 Consider the decoding at DL user $(i,j)$. Since
 \begin{align}
\left\{\mathbf{\bar{v}}_{qp}^{[{\sf bs}](a)}\right\}_{q\in[1:K], p\in[1:N], a\in[1:\frac{\lambda_2}{\lambda_1}T^{KN(KM+KN)}]~\textrm{except $q=i$ and $p=j$}}
\end{align}
is designed to null out DL intra-cell and DL inter-cell interferences, they can be removed by receive zero-forcing at DL user $(i,j)$.


Now we examine the dimension and linear independence property of the desired signals. First, notice that  $\left\{\mathbf{\bar{v}}_{ij}^{{[\sf bs]}(a)}\right\}_{a\in[1:\frac{\lambda_2}{\lambda_1}T^{KN(KM+KN)}] }$ is a set of linearly independent vectors whose cardinality is given by $\frac{\lambda_2}{\lambda_1}T^{KN(KM+KN)}$, which is constructed independent of $\mathbf{\bar{G}}_{ij,i}$. Therefore, 
\begin{align}
\left\{\mathbf{\bar{G}}_{ij,i}\mathbf{\bar{v}}_{ij}^{{[\sf bs]}(a)}\right\}_{a\in[1:\frac{\lambda_2}{\lambda_1}T^{KN(KM+KN)}] },
\end{align}
the desired DL signals for DL user $(i,j)$, is also a set of linearly independent vectors whose cardinality is given by $\frac{\lambda_2}{\lambda_1}T^{KN(KM+KN)}$. Moreover, recall that the dimension of user-to-user interferences is at most $(T+1)^{KN(KN+KM)}$ via the proposed UL IA. Since channels are generic, each of $\left\{\mathbf{\bar{G}}_{ij,i}\mathbf{\bar{v}}_{ij}^{{[\sf bs]}(a)}\right\}_{a\in[1:\frac{\lambda_2}{\lambda_1}T^{KN(KM+KN)}] }$ is linearly independent of user-to-user interferences if
\begin{align}
&\frac{\lambda_2}{\lambda_1}T^{KN(KM+KN)}+(T+1)^{KN(KN+KM)}\leq \frac{1}{\lambda_1}(T+1)^{KN(KM+KN)}
\end{align}
which is satisfied from the assumption $\lambda_1+\lambda_2\leq 1$ in \eqref{eq:Dof2}. Therefore, DL user $(i,j)$ can recover its desired DL message.

\subsubsection{Achievable DoF}
Since total $KN(1+\frac{\lambda_2}{\lambda_1})T^{KN(KM+KN)}$ submessages are delivered during $nd=\frac{n}{\lambda_1}(T+1)^{KN(KM+KN)}$ time slots via a length-$n$ codeword each, the sum DoF of
\begin{align} \label{eq:Dof2_converse}
&\frac{KN(1+\frac{\lambda_2}{\lambda_1})T^{KN(KM+KN)}}{\frac{1}{\lambda_1}(T+1)^{KN(KM+KN)}}
\end{align}
is achievable under the constraints in \eqref{eq:Dof2}. Finally, since \eqref{eq:Dof2_converse} converges to $KN(\lambda_1+\lambda_2)$ as $T$ increases, the sum DoF in \eqref{eq:Dof2} is achievable.

\subsection{Optimal $(\lambda_1, \lambda_2)$ and $\max(d_{\Sigma,1}, d_{\Sigma,2})$}
In the following, we solve the two linear programs \eqref{eq:Dof1} and \eqref{eq:Dof2} and then establish the sum DoF lower bound in Theorem \ref{thm:dof_1}.


First, consider the case where $M<(K-1)N$. The feasible $(\lambda_1, \lambda_2)$ regime of \eqref{eq:Dof1} is illustrated in Fig. \ref{Fig:feasible1} showing that $d_{\Sigma,1}=KN\left(\frac{M^2+MN}{M^2+N^2+MN}\right)$.  The feasible $(\lambda_1, \lambda_2)$ regime of \eqref{eq:Dof2} is illustrated in Fig. \ref{Fig:feasible2} showing that $d_{\Sigma,2}=\frac{MK}{(K-1)}$ for $M\leq(K-2)N$ and $d_{\Sigma,2}=\frac{2MN+M^2}{M+N}$ for $M> (K-2)N$.
Hence, for $M<(K-1)N$,
\begin{align}\label{dof_achievable1}
\max\{d_{\Sigma,1}, d_{\Sigma,2}\}=\left\{\begin{array}{lll}d_{\Sigma,1}= KN\left(\frac{M^2+MN}{M^2+N^2+MN}\right) \quad \text{if $K\geq \frac{M}{N}+1+\frac{N^2}{(M+N)^2} $},\\
d_{\Sigma,2}= \frac{2MN+M^2}{M+N} \quad\quad\quad\quad\quad~\text{if $\frac{M}{N}+1
< K\leq \frac{M}{N}+1+\frac{N^2}{(M+N)^2}$. }
\end{array} \right.
\end{align}

%
\begin{figure}[t]
\footnotesize
\begin{center}
~\subfigure[ Feasible $(\lambda_1,\lambda_2)$ of \eqref{eq:Dof1}]{%
            \label{Fig:feasible1}
            \includegraphics[width=0.3\textwidth]{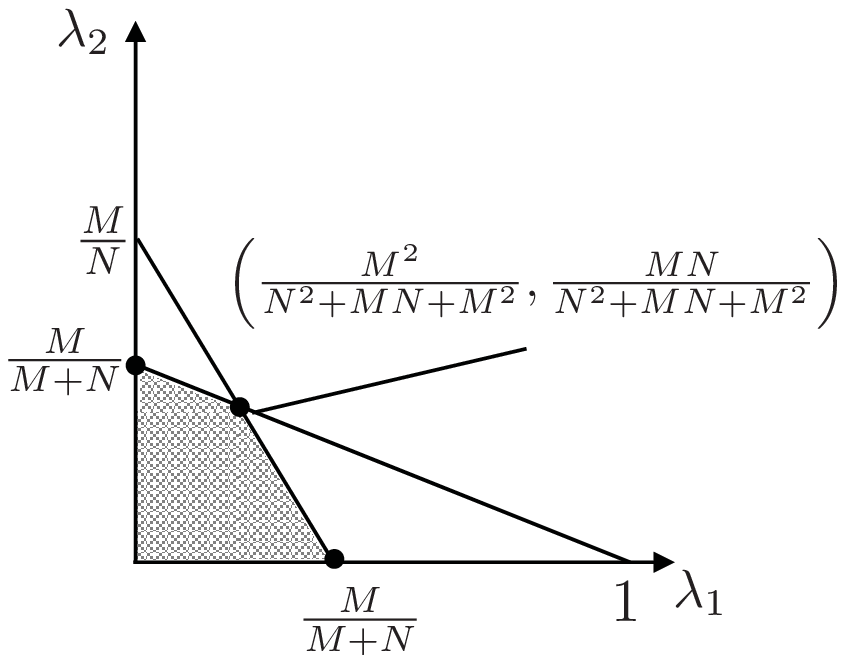}
        }~~~~~%
\subfigure[ Feasible $(\lambda_1,\lambda_2)$ of \eqref{eq:Dof2}]
{%
    \label{Fig:feasible2}
   \includegraphics[width=0.52\textwidth]{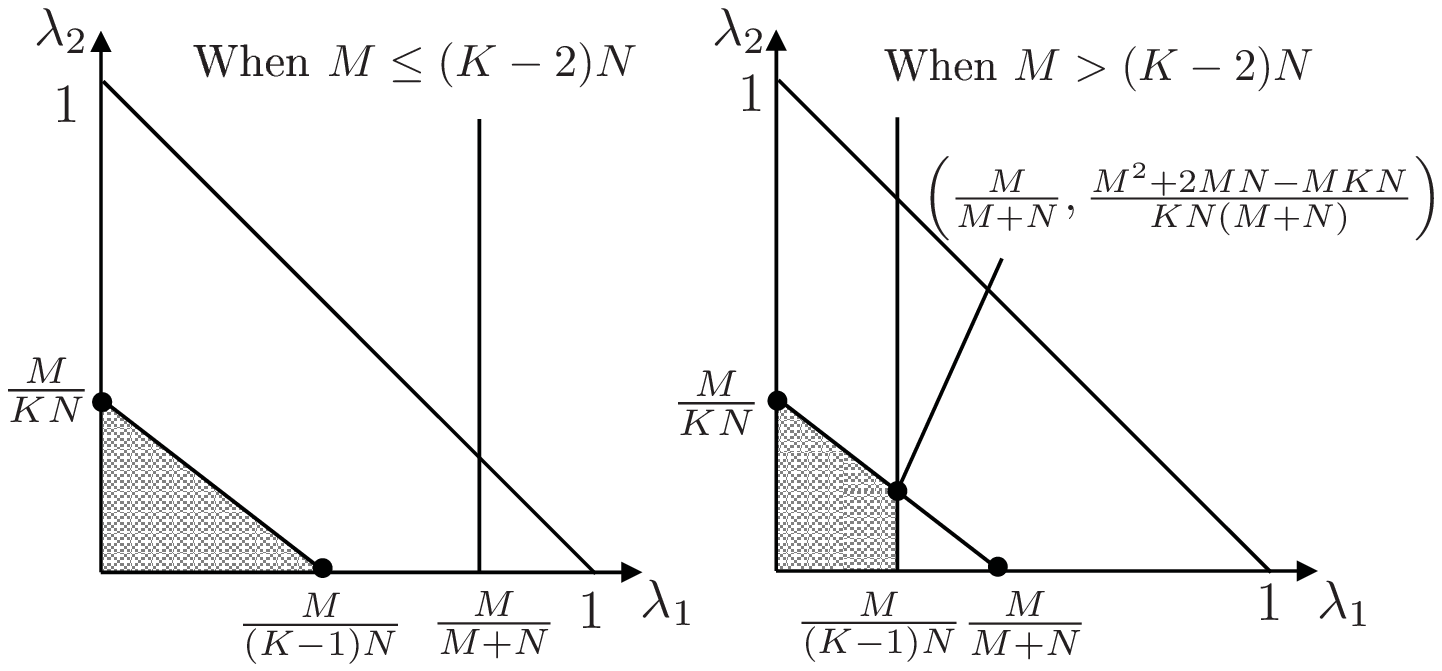}
    }
\end{center}
\caption{Feasible $(\lambda_1,\lambda_2)$  when $M>(K-1)N$.}

\end{figure}

\begin{figure}[t]
\footnotesize
\begin{center}
~\subfigure[ Feasible $(\lambda_1,\lambda_2)$ of \eqref{eq:Dof1}]{%
            \label{Fig:feasible3}
            \includegraphics[width=0.34\textwidth]{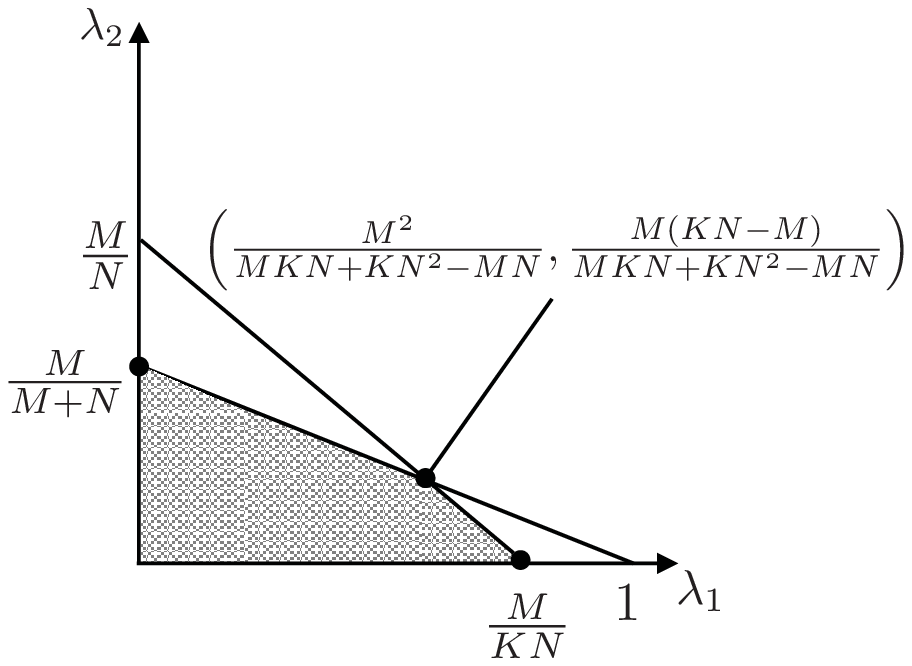}
        }~~~~~~%
\subfigure[ Feasible $(\lambda_1,\lambda_2)$ of \eqref{eq:Dof2}]
{%
    \label{Fig:feasible4}
   \includegraphics[width=0.61\textwidth]{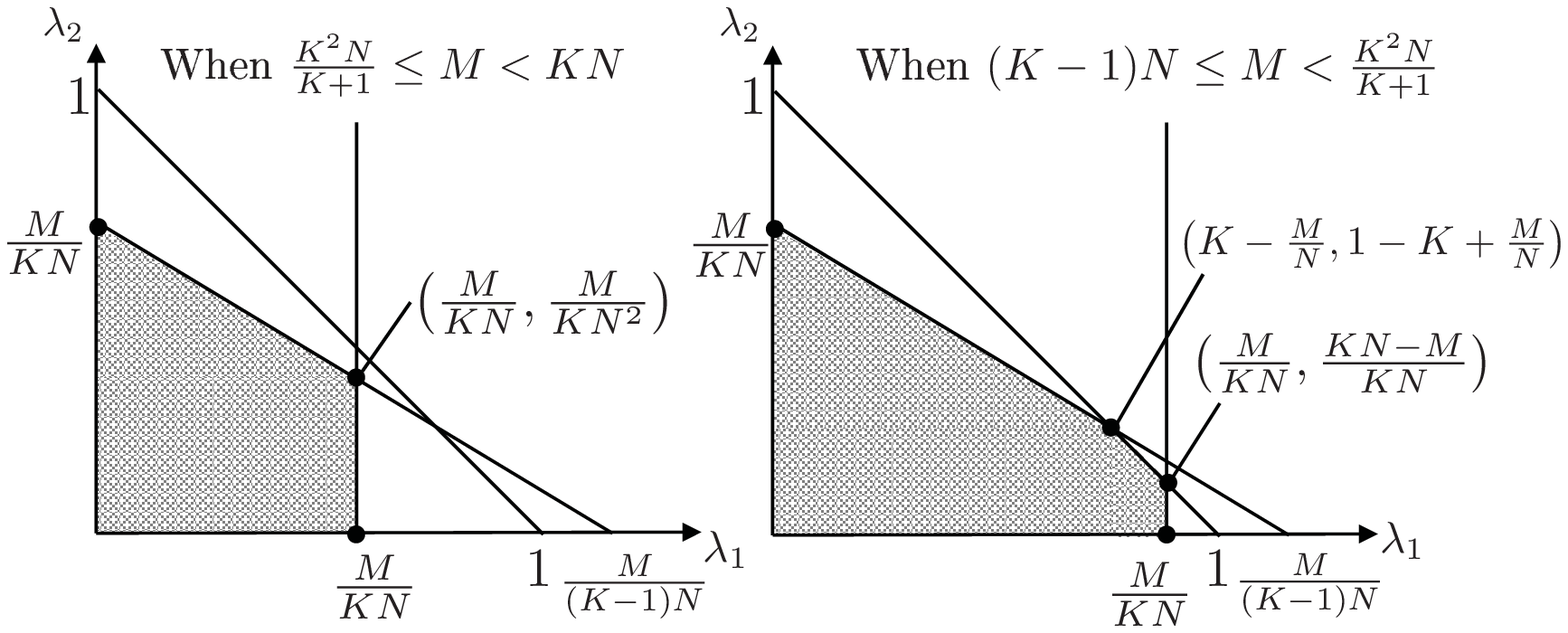}
    }
\end{center}
\caption{Feasible $(\lambda_1,\lambda_2)$  when $(K-1)N\leq M<KN$.}
\label{Fig:feasible3_4}
\end{figure}

Second, consider the case where $(K-1)N\leq M<KN$. Similarly, the feasible $(\lambda_1, \lambda_2)$ regime is given by Fig. \ref{Fig:feasible3_4} and by solving \eqref{eq:Dof1} and \eqref{eq:Dof2}, we have
\begin{align}\label{dof_achievable2}
\max\{d_{\Sigma,1}, d_{\Sigma,2}\}=d_{\Sigma,2}=\left\{\begin{array}{lll} M+\frac{M}{K}\quad \text{if $(K-1)N\leq M<\frac{K^2N}{K+1}$,}\\
~~KN~~ \quad \text{if $\frac{K^2N}{K+1}\leq M<KN$.}
\end{array} \right.
\end{align}



In summary, from \eqref{dof_achievable1}, \eqref{dof_achievable2}, and the fact that $d_{\Sigma}=KN$ for $M\geq KN$, the following sum DoF is achievable:
\begin{align}\label{dof_achievable_temp}
&\left\{\begin{array}{lll} KN\left(\frac{M^2+MN}{M^2+N^2+MN}\right) {~}\text{if $M< (K-1)N$ and $K\geq\frac{M}{N}+1+\frac{N^2}{(M+N)^2}$,}\\
\frac{2MN+M^2}{M+N}\quad\quad\quad\quad{~~}~\text{if $M<(K-1)N$ and $K\leq \frac{M}{N}+1+\frac{N^2}{(M+N)^2}$,}\\
M+\frac{M}{K} \quad\quad\quad\quad\quad~ ~\text{if $(K-1)N\leq M<\frac{K^2N}{K+1}$,}\\
~~KN~~ \quad\quad\quad\quad\quad~~ \text{otherwise.}
\end{array} \right.
\end{align}
Finally, when $(K-2)N<M<(K-1)N$, we can activate only $K-1$ cells out of $K$ cells to achieve the sum DoF of $\min\left\{{M}+\frac{M}{K-1}, (K-1)N\right\}$ from the third and fourth cases in \eqref{dof_achievable_temp}, which is greater than  
\begin{align}
\max\left\{KN\left(\frac{M^2+MN}{M^2+N^2+MN}\right),\frac{2MN+M^2}{M+N}\right\}
\end{align}
under certain conditions. Therefore, the achievable sum DoF is finally given by
\begin{align}\label{dof_achievable}
&d_{\Sigma}\geq \left\{\begin{array}{lll}
KN\left(\frac{M^2+MN}{M^2+N^2+MN}\right)~~\quad\quad\quad\quad\quad\quad {~}~\text{if $M\leq (K-2)N$, }\\
\max\left\{KN\left(\frac{M^2+MN}{M^2+N^2+MN}\right),\frac{2MN+M^2}{M+N},\min\left\{\frac{MK}{K-1}, (K-1)N\right\}\right\} \\
\quad\quad\quad\quad\quad\quad\quad\quad\quad\quad\quad\quad\quad\quad\quad\quad\text{if $(K-2)N <M< (K-1)N$,}\\
M+\frac{M}{K}\quad \quad\quad\quad\quad\quad\quad\quad\quad\quad\quad\quad~~\text{if $(K-1)N\leq M<\frac{K^2N}{K+1}$,}\\
KN~~~\quad\quad\quad\quad\quad\quad\quad\quad\quad\quad\quad\quad\quad \text{otherwise.}
\end{array} \right.
\end{align}
which completes the proof of Theorem \ref{thm:dof_1}.


\section{Discussion}\label{sec:discussion}
In this section, we briefly discuss about the impacts of BS-to-BS interference and self-interference within each BS in terms of DoF.

\begin{figure}[t!]
\begin{center}
\includegraphics[width=0.55\textwidth]{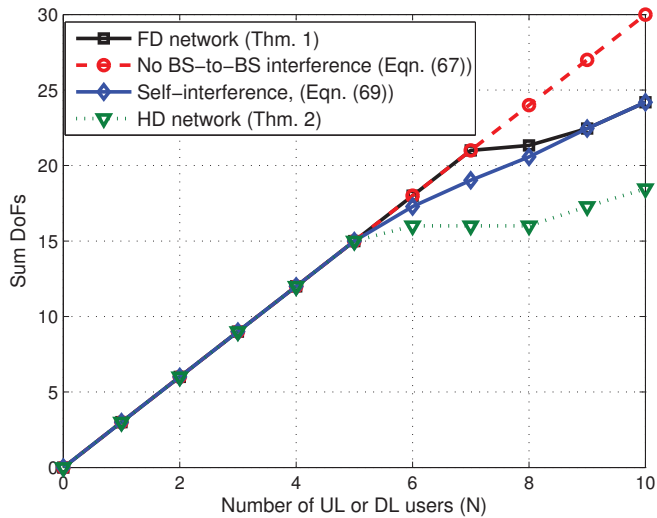}
\vspace{-0.1in}
\end{center}
\caption{Sum DoFs with respect to $N$ when $M=16$ and $K=3$.}
\label{Fig:InterBS_Self}
\vspace{-0.1in}
\end{figure}

\subsection{DoF Loss due to BS-to-BS Interferences}
Suppose that there is no BS-to-BS interference, i.e., $\mathbf{B}_{ij}=0$, $\forall i,j\in[1:K]$. In this case, the sum DoF achieved by the second proposed scheme in Section \ref{subsec:scheme2} can be improved since BSs now only need to null out DL intra-cell and inter-cell interferences. Specifically, by following a similar step, the sum DoF of
\begin{align}\label{eq:Dof2_no_inter_BS}
\max_{\substack{
\lambda_1,\lambda_2\in(0,1]\\
            \lambda_1(N+\min\{M,K(N-1)\}\leq M\\
            \lambda_2KN\leq M\\
            \lambda_1+\lambda_2\leq 1\\
           }}KN(\lambda_1+\lambda_2)
            \end{align}
is achievable in the absence of BS-to-BS interference. By solving the above linear program, the sum DoF of
\begin{align}\label{dof_achievable_no_inter_BS}
\min\left\{KN, \max\left\{M+\frac{KMN}{M+N}, 2M\right\}\right\}
\end{align}
is achievable, which is obviously greater than or equal to the lower bound in Theorem \ref{thm:dof_1}.

\subsection{Impacts of Self-Interference on DoF}
Throughout the paper, we assume no self-interference within each BS. However, for practical FD BSs, residual self-interference may exist due to imperfect or insufficient self-interference suppression and cancellation. Note that by interpreting self-interference as another additional BS-to-BS interference, the sum DoF achieved by the first proposed scheme in Section \ref{subsec:scheme2} is still achievable in the presence of self-interference by modifying~\eqref{eq:Tx_BF1_interference_DL} in Section \ref{subsec:scheme2} as
\begin{align}
\beta_{j,p_1,p_2}[t]=\left\{\begin{array}{lll} &g_{ik,lq}[t] \quad~~\text{if $p_1\leq KN$ and $i\neq l$,}\\
&g_{ik,lq}[t] \quad~~\text{if $p_1\leq KN$ and $j\neq k$,}\\
&b_{rs,lq}[t]  \quad~~ \text{if $p_1>KN$,}\\
&1  \quad\quad\quad~~~ \text{otherwise, }
\end{array} \right.
\end{align} 
where $i=\lceil\frac{p_1}{N}\rceil$, $k=p_1-N(\lceil\frac{p_1}{N}\rceil-1)$, $l=\lceil\frac{p_2}{M}\rceil$, $q=p_2-M(\lceil\frac{p_2}{M}\rceil-1)$, $r=\left\lceil\frac{p_1-KN}{M}\right\rceil$, $s=p_1-M(\left\lceil\frac{p_1-KN}{M}\right\rceil-1)$, while the sum DoF achieved by the first proposed scheme in Section \ref{subsec:scheme1} collapses to $\min\{M, KN\}$. Therefore, in this case, the following sum DoF is achievable:
 \begin{align}\label{eq:Dof1_self}
\left\{\begin{array}{lll}
KN\left(\frac{M^2+MN}{M^2+N^2+MN}\right)\quad\quad~\text{if $M\leq (K-1)N$, }\\
KN\left(\frac{MK}{MK+KN-M}\right)~\quad{~~}~\text{if $(K-1)N\leq M\leq KN$,}\\
KN~~\quad\quad\quad\quad\quad\quad\quad\quad{~}\text{otherwise,}
\end{array} \right.
\end{align}
which is obviously smaller than or equal to the lower bound in Theorem \ref{thm:dof_1}.

\subsection{Numerical Examples}
To see the impacts of BS-to-BS interference and self-interference within each BS, 
consider an example network where $K=3$ and $M=16$.
Figure~\ref{Fig:InterBS_Self} plots the sum DoFs in Theorems \ref{thm:dof_1} and \ref{thm:HD_case},  \eqref{dof_achievable_no_inter_BS}, and \eqref{eq:Dof1_self} with respect to $N$.
It is observed that when
$N$ is relatively small, the sum DoF in Theorem \ref{thm:dof_1} is the same as that for the case of no BS-to-BS interference, i.e., \eqref{dof_achievable_no_inter_BS}, because each BS is able to null out all DL intra-cell , DL inter-cell, BS-to-BS interferences for the second proposed scheme in Section \ref{subsec:scheme2}.
However, as $N$ increases, the sum DoF in Theorem \ref{thm:dof_1} eventually collapses to that for the case where there exists self-interference. i.e., \eqref{eq:Dof1_self}.
This is due to the fact that as the number of DL intra-cell and DL inter-cell interference links increases, IA at each BS (the first proposed scheme in Section \ref{subsec:scheme1}) outperforms interference nulling at each BS (the second proposed scheme in Section \ref{subsec:scheme2}) and IA at each BS can be done with or without self-interference.
More importantly, even if there exists self-interference within each BS, the proposed schemes can be straightforwardly modified by treating it as additional BS-to-BS interference and improve the sum DoF compared to the conventional HD cellular network; see Theorem \ref{thm:HD_case} and \eqref{eq:Dof1_self}. It is in contrast to the single-cell case in which FD operation at the BS does not improve the sum DoF in the presence of self-interference.


\section{Conclusion}
In this paper, we have studied the sum DoFs of multicell cellular networks consisting of multiantenna FD BS and single-antenna HD mobile users. Compared to previous works on the single-cell FD network~\cite{Jeon--Chae--Lim2015_ISIT} and the multicell HD network~\cite{Sridharan:13}, for the considered network, we additionally need to take into account for inter-cell interferences from multicell spectrum sharing and UL-to-DL interferences from FD. A novel interference management scheme has been proposed for mitigating all different types of intra-cell, inter-cell, user-to-use, and BS-to-BS interferences and the corresponding achievable DoF has been established for a general network configuration with respect to the number of antennas at each BS, the number of cells, and the number of users in each cell. For the converse part, we also derived an upper bound on the sum DoF, which is tight under certain conditions. The results show 
significant throughput improvement over conventional cellular networks, 
thereby providing theoretical evidence of performance improvement by jointly utilization of spectrum sharing and FD.

Our work can be extended to several interesting further directions. For example, regarding CSI, we may consider the case in which the channel state information at the transmitters (CSIT) is delayed or partially known in order to show whether spectrum sharing and FD can still increase the system throughput and how the DoF of the multicell FD network behaves in the lack of CSI. Another interesting direction would be extended to the case in which mobile users have multiple antennas, in order to see the tendency of DoF improvement with respect to the number antennas at the users.


\end{document}